\documentclass[conference]{IEEEtran}

\usepackage[switch, mathlines]{lineno}
\usepackage{cite}
\usepackage{amsmath,amssymb,amsfonts,amsthm}
\usepackage{algorithmic}
\usepackage{graphicx}
\usepackage{textcomp}
\usepackage{xcolor}
\usepackage{thm-restate}
\usepackage[hidelinks]{hyperref}
\usepackage{wrapfig}
\usepackage{makecell}
\usepackage{rotating}

\def\BibTeX{{\rm B\kern-.05em{\sc i\kern-.025em b}\kern-.08em
    T\kern-.1667em\lower.7ex\hbox{E}\kern-.125emX}}

\usepackage{stmaryrd}
\SetSymbolFont{stmry}{bold}{U}{stmry}{m}{n}

\theoremstyle{plain}
	\newtheorem{theorem}				 {Theorem}  
	\newtheorem{lemma}			[theorem]{Lemma}
	\newtheorem{proposition}	[theorem]{Proposition}
	
\theoremstyle{definition}
	\newtheorem{definition}		[theorem]{Definition}
    \newtheorem{remark}			[theorem]{Remark}
    \newtheorem{construction}	[theorem]{Construction}
    \newtheorem{convention}	[theorem]{Convention}
    
    \newtheorem{example}		[theorem]{Example}
    \newtheorem{notation}		[theorem]{Notation}

\usepackage{virginialake}
\usepackage{swc-commands}
\usepackage{scalerel}
\usepackage{tikz}

\makeatletter

\def\ps@IEEEtitlepagestyle{%
    \def\@oddhead{}%
    \def\@evenhead{}%
}

\makeatletter

\begin{document}

\title{Proof Compression via Subatomic Logic and Guarded Substitutions \\
\thanks{}
}

\author{\IEEEauthorblockN{Victoria Barrett}
\IEEEauthorblockA{\textit{Partout} \\
\textit{Inria Saclay}\\
Ile-de-France, France \\
}
\and
\IEEEauthorblockN{Alessio Guglielmi}
\IEEEauthorblockA{\textit{Department of Computer Science} \\
\textit{University of Bath}\\
Bath, United Kingdom \\
}
\and
\IEEEauthorblockN{Benjamin Ralph}
\IEEEauthorblockA{\textit{Department of Computer Science} \\
\textit{University of Bath}\\
Bath, United Kingdom \\
}
\and
\IEEEauthorblockN{Lutz Stra\ss burger}
\IEEEauthorblockA{\textit{Partout} \\
\textit{Inria Saclay}\\
Ile-de-France, France \\
}
}

\maketitle

\begin{abstract}
Subatomic logic is a recent innovation in structural proof theory where atoms are no longer the smallest entity in a logical formula, but are instead treated as binary connectives. As a consequence, we can give a subatomic proof system for propositional classical logic such that all derivations are strictly linear: no inference step deletes or adds information, even units.
\\
In this paper, we introduce a powerful new proof compression mechanism that we call \emph{guarded substitutions}, a variant of explicit substitutions, which substitute only guarded occurrences of a free variable, instead of all free occurrences. This allows us to construct ``superpositions'' of derivations, which simultaneously represent multiple subderivations. We show that a subatomic proof system with guarded substitution can p-simulate a Frege system with substitution, and moreover, the cut-rule is not required to do so.
\end{abstract}

This work was supported by the Inria Exploratory Action IMPROOF and PHC Sophie Germain ``Formal Verification for Large Language Models"



\vlupdate{\Tor}
\vlupdate{\Ben}
\vlupdate{\ben}

\renewcommand{\KSASG}{\TBLS}

\section{Introduction}
\label{sec:intro}

The affinity between structural proof theory and the mathematical foundations of computation establishes mechanisms of proof compression as a natural object of study. The most prominent proof compression mechanism is the \emph{cut rule}~\cite{gentzen:35:I,gentzen:35:II}, which allows lemmas to be reused in a proof. And indeed, eliminating the cuts from a proof can lead to a non-elementary blow-up~\cite{boolos:dont}. In propositional logic, this blow-up is still exponential~\cite{troelstra:schwichtenberg:00}.

In the area of proof complexity, a distinct subfield of proof theory, other mechanisms of proof compression have been studied, the most notable ones being substitution~\cite{CookReck:79:The-Rela:mf} and Tseitin extension~\cite{tseitin:68}. In both cases, proof compression is achieved by permitting propositional variables to replace arbitrary subformulae in a proof. The cost of eliminating either rule from a proof is an exponential blow-up (and it is not known whether this can be done more efficiently).

Given their conceptual similarity, it is perhaps not so surprising that, in the presence of cut, extension and substitution are \emph{p-equivalent}, i.e., a system with cut and substitution can polynomially simulate one with cut and Tseitin extension~\cite{CookReck:79:The-Rela:mf}, and vice versa~\cite{krajicek:pudlak:89}. This has been shown in the setting of Frege systems, which always contain the cut because of the presence of the \emph{modus ponens} rule. Moreover, even in the absence of cut, it has also been demonstrated that the two proof compression mechanisms of substitution and Tseitin extension are p-equivalent~\cite{str:ext-cut,nov:str:subst}. However, it is not known if cut-free systems with extension or substitution can p-simulate systems with cut and extension or substitution.

This leaves us with two powerful proof compression mechanisms: (i) the cut and (ii) extension/substitution. It is an open question whether one of them subsumes the other, or whether they are truly independent.

In this paper we give a surprising answer to this question. We observe that both proof compression mechanisms are subsumed by a more general one, namely, \emph{guarded substitution}, which is a variant of explicit substitutions~\cite{curien:explicit}.

To see how this is possible, let us first observe that the proof compression of cut and substitution comes from the ability of reusing information. And the most basic inference rules that deal with duplication of information are the rules of \emph{contraction} and \emph{cocontraction} shown below:\footnote{In fact, the combination of contraction and cocontraction form another mechanism of proof compression, when cut is absent. This has been investigated in~\cite{jerabek:cut-free-cos,BGGP:qp-lmcs,Das:13:The-Pige:fk}. However, we will not go into further details of this, as we present in this paper a cut-free system that can p-simulate cut.}
\begin{equation}
  \label{eq:con-cocon}
  \odn{A\vlor A}{\cD}A{}
  \qquand
  \odn{A}\cU {A\vlan A}{}
\end{equation}

A move to a deep inference proof system~\cite{gug:str:01,brunnler:tiu:01}, which allow for finer granularity in the design of the inference rules, enables the restriction of these rules to their atomic form, shown below, provided there is an additional purely linear inference rule in the system~\cite{brunnler:tiu:01}.
This rule, called \emph{medial}, is shown in the middle below:
\begin{equation}
  \label{eq:acon-acocon}
  \hskip-1em
  \odn{a\vlor a}{\acD}a{}
  \hskip2em
  \odn{\vls[(A.B).(C.D)]}{\med}{\vls([A.C].[B.D])}{}
  \hskip2em
  \odn{a}\acU {a\vlan a}{}
  \hskip1em
\end{equation}
This leads to the proof system $\SKS$~\cite{brunnler:tiu:01}, consisting only of atomic rules (like $\acD$ and $\acU$ above) that change the size of a formula, and purely linear rules (like $\med$ above) that only rearrange subformulae without changing the size.

\begin{figure*}[!t]
  $$
  \vliiinf{\cut}{}{\sqn{\Gamma,\Delta}}{\sqn{\Gamma,A}}{}{\sqn{\cneg A,\Delta}}
  \qualto
  \odn{\,A\vlan \cneg A\,}{\iU}{\zer}{}
  \qualto
  \odn{\,a\vlan \cneg a\,}{\aiU}{\zer}{}
  \qualto
  \odn{((\zer\ba\one). (\one\ba\zer))}{}{((\zer.\one)\ba(\one.\zer))}{}
  \qualto
  \odn{((A\ba B). (C\ba D))}{\sUp\ba\vlan}{((A.C)\ba(B.D))}{}
  $$
  \caption{Evolution of $\cut$ from the sequent calculus via deep inference to subatomic proof theory}
  \label{fig:cut}
\end{figure*}

The next insight comes from the concept of \emph{subatomic proof theory}~\cite{Aler:16:A-Study-:hc,BarrGugl:22:A-Subato:yr,BarrGuglRalph:25:strlin} which splits the atoms into binary connectives. A formula ``$A\ba B$'' is then interpreted as ``if $a$ is false then $A$, and if $a$ is true then $B$''. In this setting, we can write $a$ as $\zer\ba\one$ and its dual $\cneg a$ as $\one\ba\zer$. The two rules of $\acD$ and $\acU$  from above now become:
\begin{equation}
  \label{eq:acon-acocon-unit}
  \!
    \odn{\vls[(\zer\ba\one).(\zer\ba\one)]}{}{\vls([\zer.\zer]\ba[\one.\one])}{}
    \qquand
    \odn{\vls[(\zer.\zer)\ba(\one.\one)]}{}{\vls((\zer\ba\one).(\zer\ba\one))}{}{}
    \quad
\end{equation}
which are just instances of the general rules
\begin{equation}
  \label{eq:acon-acocon-gen}
  \!
    \odn{\vls[(A \ba B).(C \ba D)]}{\sDn\vlor\ba}{\vls([A . C]\ba[B. D])}{}
    \quand
    \odn{\vls[(A. B)\ba(C.D)]}{\sUp\vlan\ba}{\vls((A\ba C).(B\ba D))}{}
    \;
\end{equation}
which have the same shape as the medial rule in~\eqref{eq:acon-acocon} above.
The same principle also applies to the cut rule: in Figure~\ref{fig:cut} we see the evolution of the cut, starting from the sequent calculus, first becoming an atomic rule, and finally a subatomic rule. In this way we can obtain a proof system for propositional logic in which \emph{all} inference rules are linear rewriting steps~\cite{Aler:16:A-Study-:hc}, except for the rules dealing with the units, for example
\begin{equation}
  \label{eq:units}
  \hskip-1em
  \odn A{}{A\vlan \one}{}
  \qqquad
  \odn A{}{A\vlor \zer}{}
  \qqquad
  \odn {A\vlan \one}{}A{}
  \qqquad
  \odn {A\vlor \zer}{}A{}
  \quad
\end{equation}
Even though these are rather trivial inference steps in a standard proof system, they are the only ones that break a rigidly defined notion of linearity in a subatomic proof system. To achieve what is called a strictly linear system, these can be eliminated, but the naive way of doing so leads to an exponential blow-up of the size of the proof~\cite{Barr:24:thesis,BarrGuglRalph:25:strlin}. However, by allowing explicit substitutions as constructors in formulae and derivations, the size of the proof expansion can be reduced to a polynomial \cite{BarrGuglRalph:25:strlin}. The resulting system (called $\KDTS$ in~\cite{BarrGuglRalph:25:strlin}) is p-equivalent to $\SKS$ (and therefore also to standard Frege systems without extension or substitution) and still contains the cut (in its linear form, as shown on the right in Figure~\ref{fig:cut}). And, unsurprisingly, eliminating the cut from this system leads to an exponential blow-up~\cite{BarrGuglRalph:25:strlin}. Furthermore, it is unknown whether these explicit substitutions can in any way polynomially simulate Tseitin extension or substitution in Frege systems.

In other words, in terms of proof complexity, nothing has been gained so far with respect to what we said at the beginning of this introduction; even in the unfamiliar climes of a subatomic proof system with explicit substitutions, it appears that both cut and extension/substitution operate as independent means of compressing proofs.

This motivates our paper, in which we introduce \emph{guarded substitutions} that offer us a distinct new means of proof compression.

Guarded substitutions are a variant of explicit substitution that, instead of representing the replacement of every occurrence of a free variable, only select a certain subset of the free occurrences of the given variable. To make this formal, we assign to every variable occurrence a \emph{range}, and to every guarded substitution a \emph{guard}, and the substitution can apply to the variable occurrence, if the guard is in the range. For example, the formula $\esbs A x ((x\vlan x)\vlor(y\vlan x))$ with an ordinary explicit substitution becomes $(A\vlan A)\vlor(y\vlan A)$ when the substitution is carried out, whereas the formula $\gsbs A x p ((x\rng{q,p}\vlan x\rng r)\vlor(y\rng p\vlan x\rng {p,r}))$ becomes $(A\vlan x\rng r)\vlor(y\rng p\vlan A)$ when the substitution is carried out, because only the first and the last $x$ have the guard $p$ in its range.

With this additional construct, the new system, that we call $\TBLS$, can polynomially simulate Frege systems with substitution. Furthermore, we can do this with the cut-free fragment $\TBLScf$. In other words, guarded substitution can polynomially simulate the cut, as well as Tseitin extension and substitution in Frege systems. And surprisingly, this can even be done when we only allow the units $\zer$ and $\one$ in the place of the formula to be substituted (the $A$ in the example above).

Since Frege systems with substitution are known to be the most potent propositional proof systems, in the sense that they can polynomially simulate every other proof system for classical propositional logic, our system $\TBLS$ has now the same property, with the additional feature that every inference step is linear. There is never an inference that adds or deletes information.\looseness=-1

This becomes possible because a subatomic derivation can be interpreted as a \emph{superposition} of standard derivations. Even though this idea has been around since the beginning of subatomic proof theory~\cite{guglielmi:AG8}, only our guarded substitutions allow to make use of this. We can see such a  superposition as executing several similar shaped derivations in parallel, and the guarded substitutions can be used to read out the correct results. This allows for the factorisation of lemmas with different inputs, essentially performing the work of both modus ponens and the Frege substitution rule.


This paper is organized as follows.
\begin{itemize}
\item In the next section we present our system~$\TBLS$. We introduce the principles of subatomic proof theory, as presented in~\cite{Aler:16:A-Study-:hc,BarrGugl:22:A-Subato:yr}, together with the explicit substitutions of~\cite{BarrGuglRalph:25:strlin}, and our guarded substitutions.
\item Then, in Section~\ref{sec:sub-to-standard}, we show how subatomic proof theory is related to standard proof theory.
\item In Section~\ref{sec:cut}, we discuss the cut and cut elimination in subatomic systems, and we formalize the notion of \emph{superposition}.
\item In Section~\ref{sec:frege}, we formally introduce Frege systems with substitution.
\item And finally, in Section~\ref{sec:psim}, we prove our main result, showing that the cut-free fragment of our system~$\TBLS$ polynomially simulates Frege systems with substitution.
\item We end with a discussion in Section~\ref{sec:guarded} on the possible use of guarded substitutions beyond the restricted setting of this paper.
\end{itemize}

\section{A subatomic system with guarded substitution}
\label{sec:prelims}

\subsection{Subatomic formulae}

We begin our exposition by recalling the basic notions of subatomic proof theory~\cite{AlerGugl:18:Subatomi:ty,BarrGugl:22:A-Subato:yr} with explicit substitution, as introduced in~\cite{BarrGuglRalph:25:strlin}. At the same time, we introduce \emph{guarded substitutions}, for which we need a notion of guards that will be attached to variables.

\begin{definition}
We start from the following countable sets:
\begin{itemize}
\item \defn{Variables}: $\Vbls=\set{x,y,z,w,\dots}$,
\item \defn{Atoms}: $\Atms=\set{a,b,c,\dots}$,
\item \defn{Connectives}: $\Conn=\set{\vlan,\vlor}$,
\item \defn{Units}: $\Unts=\set{\zer,\one}$,
\item \defn{Guards}: $\Grds=\set{p,q,r,\ldots}$.
\end{itemize}
A set of guards is called a \defn{range}, and we denote the set of ranges with $\Rngs=\powerset{\Grds}=\set{R,S,\ldots}$.

The set $\Fmla=\set{A,B,C,\ldots}$ of \defn{(subatomic) formulae} is defined as follows:
\[\begin{array}{rcl}
  \Fmla &\grammareq & \Vbls^\Rngs \mid \Unts \mid \Fmla \mathbin{\Atms} \Fmla \mid \Fmla \mathbin{\Conn} \Fmla \mid \esbs\Fmla \Vbls \Fmla \mid \gsbs\Unts \Vbls \Grds \Fmla 
\end{array}\]
On the set $\Fmla$ we define the \defn{negation function} $\cneg\cdot\colon\Fmla\to \Fmla$ inductively as follows:
\[\begin{array}{r@{\;=\;}l@{\qquad\;}r@{\;=\;}l@{\qquad\;}r@{\;\;}l}
  \ccneg {x\rng R}&x\rng R
  &
  \ccneg {A\ba B} &\ccneg A \ba \ccneg B
  &
  \\[1ex]
  \ccneg \zer &\one
  &
  \ccneg {A\vlan B} &\ccneg A \vlor \ccneg B
  &
  \\[1ex]
  \ccneg \one&\zer
  &
  \ccneg {A\vlor B} &\ccneg A \vlan \ccneg B
\end{array}\]
It follows that $\ccneg{\ccneg B}=B$.
\end{definition}

\begin{notation}
Every variable occurrence in a formula $A$ is equipped with a range
$R$, which is a set of guards. We omit the range when it is
empty, writing $x$ instead of $x^\emptyset$. 
\end{notation}

\begin{notation}
In the subatomic setting, atoms do not
behave as atomic formulae or nullary connectives, but as
binary connectives. We make this distinction by writing atomic connectives in bold: $\ba,\bb,\bc,\dots$. The intended semantics of a formula $A\ba B$ is $A$ whenever $a$ is false and $B$ whenever $a$ is true. This means we can recover the standard (i.e. non-subatomic) atom $a$ with $\zer\ba\one$ and
its negation $\cneg a$ with $\one\ba\zer$.\footnote{We will make this formal in the next section.}
\end{notation}

\begin{definition}
The term $\esbs A y$ is an \defn{explicit substitution} and the term $\gsbs A y p$ is a \defn{guarded (explicit) substitution}. If $A$ and $B$ are formulae then both $\esbs A yB$ and $\gsbs A y pB$ are formulae, and in both cases, every occurrence of $y$ in $B$ is \defn{bound}. Every variable occurrence that is not bound is \defn{free}. We denote the set of free variables of a formula $A$ by $\fv{A}$. We do not consider the substitution variable $y$ in $\esbs A y$ or in $\gsbs A y p$ to be an occurrence of the variable $y$, so it is neither free nor bound. A \defn{flat} formula contains neither explicit nor guarded substitutions.  
\end{definition}

For guarded substitutions, we restrict the formula $A$ that is substituted to be a unit: this restriction still provides us with a system strong enough for all our present purposes, while simplifying many definitions and proofs.

\begin{remark}
We write
$\range A S$ to denote the formula obtained from $A$ by adding $S$ to the range of every free variable occurrence in $A$, i.e., $\range {x^R}S= x^{R\cup S}$. We often omit the braces to improve readability, for example writing $\range A r$ for $\range A {\set r}$ and $x\rng {p,r}$ for $x^{\set {p, r}}$.
\end{remark}

\begin{definition}
    We write $\asbs A y$ to denote an \defn{actual substitution}, that is, $\asbs A y B$ denotes the formula that is obtained from $B$, by replacing every free variable occurrence $y^S$ in $B$ by $\range A S$. We use $\agsbs A y p B$ to denote a \defn{guarded actual substitution}, that is, $\agsbs A y p B$ denotes the formula that is obtained from $B$ where every free occurrence $y^S$ in $B$ is replaced by $A$ whenever $p\in S$, and remains untouched whenever $p\not\in S$.\footnote{As this will only be used when $A$ is a unit, we substitute just with $A$ instead of $\range A S$ as in the general case.} 
\end{definition}

\begin{notation}\label{n:simultaneaous} 
We use lower case Greek letters $\sigma,\nu,\tau,\rho,\eta,\mu$ to denote substitutions. Furthermore, if $B$ is a formula, and $V=\set{v_1,v_2,\ldots,v_n}$ is a finite set of variables, and $A_1,\ldots,A_n$ is a list of formulae, then $\acsnoun {A_i}{v_i} V B$ stands for the \emph{simultaneous} actual substitution of all free occurrences of the variables $v_1,\ldots, v_n$ in $B$ for $A_1,\ldots,A_n$, respectively. In the same way we can define $\exsnoun {A_i}{v_i} V B$ as the simultaneous explicit substitution of $\esbs{A_1}{v_1}\ldots\esbs{A_n}{v_n}$. But note that in the case of explicit substitution we need that none of the $v_i$ occurs freely in any of the $A_j$.

We often write $\acs {A_i}{v_i} B B$ or $\exs {A_i}{v_i} B B$ to indicate the simultaneous substitutions of all the free variables of $B$. We often encounter guarded substitutions in pairs, and we write $\gsbss A x B y p$ as abbreviation for $\gsbs A x p \gsbs B y p$. Finally,  we allow for empty explicit and guarded substitutions, denoting them both by $\emptysub$, and then we have $\emptysub A=A$ for all $A$.
\end{notation}

\begin{definition}
  The \defn{size} of a formula $A$, denoted by $\sizeof{A}$ is defined inductively as follows:
  For $A\in\Unts$ we have $\sizeof{A}=1$, for $A^R\in\Vbls^\Rngs$ we have $\sizeof{A}=1+\sizeof{R}$ , for $\conna\in\Atms\cup\Conn$ we have $\sizeof{B\conna C}=\sizeof{B}+\sizeof{C}$, and we have $\sizeof{\esbs B y C}=\sizeof{\gsbs B y p C}=\sizeof{B}+\sizeof{C}$.
\end{definition}

\begin{example}
Let $A\ideq \gsbs \zer x r \esbs{x \ba \zer}{y}(y\rng{p,q}\vlor y\rng{p,r})$, where $x,y\in\Vbls$ and $p,q,r\in\Grds$, and $\ba\in\Atms$. The variable occurrence~$x$ is bound by the guarded substitution $\gsbs\zer x r$ and the variable occurrences $y\rng{p,q}$ and~$y\rng{p,r}$ are bound by the explicit substitution $\esbs {x \ba \zer}{y}$, and we have $\sizeof{A} = 9$.
\end{example}


\subsection{Subatomic derivations}

In this paper, we use the open deduction framework of deep inference \cite{GuglGundPari::A-Proof-:fk}, in which derivations are built with the same constructors that are used for formulae with additional constructors for vertical composition.

\begin{definition}\label{def:subatomic-pre-derivations}
The set $\SADerv=\set{\dera,\derb,\derc,\dots}$ of \defn{(subatomic) pre-derivations}, is defined by the grammar:
\[
\begin{array}{rcl}         
  \SADerv & \grammareq & \Vbls^\Rngs\mid\Unts  \mid    \SADerv \mathbin{\Atms} \SADerv \mid    \SADerv \mathbin{\Conn} \SADerv \mid
  \esbs{\SADerv}\Vbls\SADerv \mid    \\ \\[-1ex]
  &&  
            \gsbs\Unts \Vbls \Grds \SADerv \;\; \left\vert\ \;\; \vls\odn{\;\SADerv\;}{}{\SADerv}{} 
            \;\;\!
            \right. \left\vert\;\;
            \vls\odq{\;\SADerv\;}{}{\SADerv}{}
            \right.           
\end{array}
\]
The notions of bound and free variables carry over from formulae to derivations. In particular, in $\esbs \dera y \derb$, every occurrence of $y$ in $\derb$ is \defn{bound}, the substitution variable $y$ in $\esbs \dera y$ is not considered an occurrence of the variable $y$, and every variable occurrence that is not bound is \defn{free}. Similarly, a derivation is \defn{flat} if it contains no explicit and no guarded substitution. We say that a derivation is \defn{open} if it contains no units, i.e., every leaf in the syntax tree is a free variable that can be substituted into; and we say that a derivation is \defn{closed} if it contains no free variables, i.e., no leaf in the syntax tree can be substituted into.
\end{definition}

We note that because guarded substitutions always substitute units, no open derivation can contain guarded substitutions.

We can define the two maps \defn{premise} and \defn{conclusion}, ${\pr},{\cn}\colon\SADerv\to\Fmla$, and the two maps \defn{width} and \defn{height}, ${\w},{\h}\colon\SADerv\to\Nat$ in the standard way.\footnote{The detailed definition is in the appendix.}

\begin{example}\label{ex:prederi}
  Below are two pre-derivations:
  \begin{equation}
    \label{eq:deri1}
    \dera\ideq\vls \odn{
      ((x\ba z).
      \odq{\odframefalse\esbs{\odn{(w.x)}{\mix}{[w.x]}{}}{y}(y\rng p\ba y\rng q) }
          {}{([w\rng p.x\rng p]\ba[w\rng q.x\rng q] )}{})}
                   {\sUp{\vlan}\ba}{((x.[w\rng p.x\rng p]) \ba (z.[w\rng q.x\rng q]))}{} 
  \end{equation}
  and\padjust{-2ex}
  \begin{equation}
        \label{eq:deri2}
        \derb\ideq \odn
                   {((\zer\ba\one).
                     (\odn{(w.\one)}{\mix}{[w.\one]}{}\ba
                     \odn{(w.\zer)}{\mix}{[w.\zer]}{} ) )}
                   {\sUp{\vlan}\ba}{((\zer.[w.\one]) \ba (\one.[w.\zer]))}{}
  \end{equation}
  The first, $\dera$, is open, because it contains no units. And the second, $\derb$, is flat because it contains no explicit or guarded substitutions. The reader can ignore for now the labels at the vertical compositions, which will be defined below.

The width of $\dera$ is $10$ and its height is $4$. Its premise is $(x\ba z)\vlan \esbs{w\vlan x}{y}(y\ba y)$ and its conclusion is $(x\vlan (w\vlor x))\ba (z\vlan (w\vlor x)) $. The width of $\derb$ is $6$ and its height is $3$. Its premise is $(\zer\ba\one)\vlan ((w\vlan \one)\ba(w\vlan \zer))$ and its conclusion is $(\zer\vlan(w\vlor\one))\ba(\one\vlan (w\vlor\zer))$.
\end{example}

In order to define derivations, i.e.\ correct pre-derivations, we need to put restrictions on each mode of vertical composition. The first, $\downsmash{\odn{\;\SADerv\;}{}{\SADerv}{}}$ is the familiar \defn{composition by inference}, the application of an inference rule. The second, $\odq{\;\SADerv\;}{}{\SADerv}{}$, is called \defn{composition by expansion}, representing the application of an explicit or guarded substitution.

In the remainder of this section, we define this formally. As said before, in the subatomic setting, atoms are not the smallest entity of a derivation but behave like binary connectives. This is made explicit in the definition of the inference rules.

\begin{definition}
On the set $\Atms\cup\Conn$, we define the two operations \defn{down-saturation} $\sDn\cdot$ and \defn{up-saturation} $\sUp\cdot$ as follows: $\sDn{\vlan}=\sDn{\vlor}=\vlor$ and $\sUp{\vlan}=\sUp{\vlor}=\vlan$ and $\sUp{\ba}=\sDn{\ba}=\ba$ for every $\ba \in \Atms$. These operations can be extended to formulae: $\sDn A$ and $\sUp A$ are obtained from a formula $A$ by saturating all connectives down or up respectively, and leaving all variables and units unchanged. 
\end{definition}

\begin{definition}
A \defn{subatomic rule scheme} is any of the following:
\begin{equation}\label{eq:schemes}
\begin{array}{c@{\qquad}c}
          \odn{((x\mathbin{\sUp\alp}y)\bet(z\alp w))}
{\sUp\alp\bet}{((x\bet              z)\alp(y\bet w))}{}&
          \odn{((x\alp y)\bet(z\mathbin{\sUp\alp}w))}
{\bet\sUp\alp}{((x\bet z)\alp(y\bet              w))}{}\\
\\[-1ex]
          \odn{((x\bet              y)\alp(z\bet w))}
{\sDn\alp\bet}{\vlsmallbrackets
               ((x\mathbin{\sDn\alp}z)\bet(y\alp w))}{}&
          \odn{((x\bet y)\alp(z\bet             w))}
{\bet\sDn\alp}{\vlsmallbrackets
               ((x\alp z)\bet(y\mathbin{\sDn\alp}w))}{}\\
\end{array}
\end{equation}
where $\alp,\bet\in\Atms\cup\Conn$.
We call any inference rule generated by this scheme a \defn{subatomic rule}.

We also use in this paper the following inference rules, which we call the \defn{algebraic rules}:
\begin{equation}\label{eq:extra-rules}
  \begin{array}{c}
    \odn{(x.[y.z])}{\sw}{[(x.y).z]}{}
    \qquad
    \odn{[x.[y.z]]}{\assD }{[[x.y].z]}{}
    \qquad
    \odn{(x.(y.z))}{\assU }{((x.y).z)}{}
    \\\\[-1ex]
    \odn{x \vlan y}{\mix}{x \vlor y}{}
    \qquad
    \odn{x \vlor y}{\sDn\com}{y \vlor x}{}
    \qquad
    \odn{x \vlan y}{\sUp\com}{y \vlan x}{}
  \end{array}
\end{equation}
We define the proof system $\TBLS$ to be the set of all subatomic and algebraic rules.\footnote{The reader familiar with subatomic systems might observe that the algebraic rules are not needed. We will come back to this point at the end of the next section.}
\end{definition}

\begin{remark}
  Note that $\TBLS$ is a \defn{strictly linear system}, which does not have any inference rules that can create or delete information, even units \cite{BarrGuglRalph:25:strlin}.
\end{remark}

\begin{definition}
  An \defn{instance} of an \defn{inference rule} $\downsmash{\odn{A}{\rr}{B}{}}$ is given by ${\odn{\eta A}{\rr}{\eta B}{}}$, where $\eta$ is an arbitrary actual substitution. A \defn{proof system} $\Ssys$ is a set of inference rules. We say that the composition by inference of two pre-derivations $\dera$ and $\derb$, 
  $\vls\odn{\;\dera\;}{}\derb{}\,$ is \defn{correct (in $\Ssys$)} if
  $\odn{\;\cn\dera\;}{\rr}{\pr\derb}{}$ is an instance of an inference rule in $\rr \in \Ssys$, and we write this as $\upsmash{\vls\odn{\;\dera\;}{\rr}\derb{}}$.
\end{definition}

\begin{definition}
  On formulas, we define two binary \defn{expansion relations} $\expeq$ on and $\expto$.
  \begin{equation}
    \label{eq:expes}
    \begin{array}{rcll}
       \sigma(A \alp B)&\expeq&\sigma A \alp \sigma B\\[.5ex]
      \esbs C x \esbs A y B &\expeq& \esbs {\esbs C x A} y \esbs C x B\\[.5ex]
      \esbs A y D&\expeq&\asbs A y D\\[1.3ex]
    \end{array}
  \end{equation}
  where $\sigma$ is a guarded or explicit substitution, and $\alp\in\Atms\cup\Conn$, and $A,B,C$ are arbitrary formulas, and $D$ is a formula that does not contain any guarded substitutions. 
\begin{equation}
    \label{eq:expgs}
    \begin{array}{rcll}
      \gsbs U y p F&\expto&\agsbs U y p F\\[.5ex]
      \gsbs U y p \esbs F x E&\expto&\esbs{\agsbs U y p F}x E \\[.5ex]
      \gsbs U y p E&\expto&E 
    \end{array}
  \end{equation}
where $U\in\Unts$, and $F$ is a flat formula, and $E$ is a formula such that either $y$ does nor occur in $E$ (bound or free), or such that every variable occurrence in $E$ (bound or free) has a range $R$ with $p\notin R$.
\end{definition}

\begin{definition}
  Let $\dera$ and $\derb$ be pre-derivations. Their composition by expansion $\vls\odq{\;\dera\;}{}\derb{}\,$ is \defn{correct} iff
  one of the following
  $$\cn\dera\expto\pr \derb
  \quor
  \cn\dera\expeq\pr \derb
  \quor
  \pr \derb\expeq\cn\dera
  \quadfs
  $$
\end{definition}
\begin{definition}
  Let $\Ssys$ be a set of inference rules. A \defn{derivation in $\Ssys$} (or \defn{$\Ssys$ derivation}) is a pre-derivation in which every composition by inference is correct in $\Ssys$ and every composition by expansion is correct. We write $\odv{A}{\Phi}{B}{\Ssys}$ for a derivation $\Phi$ in $\Ssys$, where $\Ssys$ may be given as a rule scheme such as $\sUp\alp\bet$. If $\Ssys$ is empty then every vertical composition in $\Phi$ is by expansion; in this case we may instead annotate the derivation with its height.
\end{definition}

\begin{example}
  The two pre-derivations $\dera$ and $\derb$ given in Example~\ref{ex:prederi} are derivations in $\TBLS$, with the correct rule names indicated at the compositions by inference.
\end{example}
\begin{remark}
Our main contribution in this paper is to show that a strictly linear system equipped with explicit substitution of derivations and guarded substitutions of units can share and redistribute logical material very efficiently.
 
For this reason, the definition of the expansion relations $\expto$ and $\expeq$ is more restrictive than necessary. For example, guarded substitutions can only be applied downwards in a derivation, and have a weaker distributive relation than explicit substitutions, which is all that we need to prove the main result in this paper. This restrictive definition makes it easy to see that the time complexity of checking correctness of a derivation is quadratic in the size of the derivation. Usually, in deep inference systems this is linear, because all rules are given either linearly or atomically. In $\TBLS$ it is linear for almost all cases, except for the composition by expansion of $\esbs A y D \expeq \asbs A y D$, which is quadratic because we additionally need to check that the ranges of the free occurrences of $y $ in $D$ are correctly inherited by the free variables in $A$. If we allow guarded substitutions to freely apply to formulae containing explicit substitutions, then it is not immediately obvious how to check correctness in polynomial time.
\end{remark}

\subsection{Subatomic Merge}

We end this section by presenting two crucial lemmas that are needed
in many subatomic constructions, and that we call \emph{merge lemmas}.
The first is a generalization of the \emph{merge contraction} of~\cite{Ralp:19:Modular-:yq} and of the generalized \emph{medial map} of~\cite{lamarche:gap}.

\begin{lemma}[restate = MergeMedial, name = ]
  \label{lemma:mergemedial}
Let $A$ be an open formula with $\fv A=\set{x_1,\ldots,x_n}$. Then for all $\alp\in\Atms\cup\Conn$ and $B_1,C_1,\ldots,B_n,C_n\in\Fmla$, there exist derivations
  \[
\odv{\acs{B_i}{x_i}{A}A \mathbin{\sDn\alp} 
     \acs{C_i}{x_i}{A}A }
{}{\acs{B_i\mathbin{\sDn\alp} C_i}{x_i}{A}A}{\sDn\alp\bet,\bet\sDn\alp} \qquad\quad
\odv{\acs{B_i\mathbin{\sUp\alp} C_i}{x_i}{A}A}
{}{\acs{B_i}{x_i}{A}A \mathbin{\sUp\alp} 
     \acs{C_i}{x_i}{A}A}{\sUp\alp\bet,\bet\sUp\alp}
\]
The height of each of these derivations is $O(\size A)$ and their width is $O(\sum (\size {B_i} + \size {C_i}))$.
\end{lemma}

The second merge lemma is a generalization of the construction that is usually used to show that in a well-designed system (sequent calculus or deep inference) the general identity axiom can be reduced to an atomic form.

\begin{lemma}[restate = MergeSwitch, name = ]
  \label{lemma:mergeswitch}
Let $A$ be an open formula with $\fv A=\set{x_1,\ldots,x_n}$. Then for all $\alp\in\Atms\cup\Conn$ and $B_1,C_1,\ldots,B_n,C_n\in\Fmla$, there exist derivations 
\[
\odv{\acs{B_i \alp C_i }{x_i}{A} \sUp A}
{}{\acs{B_i}{x_i}{A} A \alp
  \acs{C_i }{x_i}{A} \cneg A}{\sDn\bet\alp,\alp\sDn\bet}
\qquad\quad
\odv{\acs{B_i}{x_i}{A} A \alp
   \acs{C_i }{x_i}{A} \cneg A}
{}{\acs{B_i \alp C_i }{x_i}{A} \sDn A}
{\sUp\bet\alp,\alp\sUp\bet}
\]
The height of each of these derivations is $O(\size A)$ and their width is $O(\sum (\size {B_i} + \size {C_i}))$.
\end{lemma}

\section{From subatomic proofs to standard proofs}
\label{sec:sub-to-standard}

\subsection{Interpretation map}

For translating subatomic formulae and derivations into standard formulae and derivations, we use \defn{interpretation maps}. These are partial functions from subatomic derivations to standard derivations. They are partial, as not every subatomic formula has a standard correspondence. 

While interpretation functions also exist for linear logics \cite{Aler:16:A-Study-:hc}, in this paper we are solely concerned with classical propositional logic and therefore provide an interpretation function from the subatomic system with guarded substitutions $\TBLS$, defined in the previous section, to the standard deep-inference proof system $\SKS$~\cite{brunnler:tiu:01,Brun:04:Deep-Inf:rq,GuglGundPari::A-Proof-:fk}. 
 
\begin{definition}
  We first define the set $\cneg\Atms=\set{\cneg a,\cneg b,\cneg c,\ldots}$ of \defn{negated atoms},
  and an in operation $\cneg\cdot\colon\Atms\to\cneg\Atms$, mapping each atom $a\in\Atms$ to $\cneg a\in\cneg\Atms$.

  We can now define the set $\StFmla=\set{A,B,C,\ldots}$ of \defn{standard formulae},  and the set
  $\StDerv$ of \defn{standard pre-derivations} as follows:
\[\begin{array}{rcl}
  \StFmla &\grammareq &  \Unts \mid \Atms \mid \cneg\Atms \mid \StFmla \mathbin{\Conn} \StFmla 
\\[1ex]
  \StDerv &\grammareq &  \Unts \mid \Atms \mid \cneg\Atms \mid \StDerv \mathbin{\Conn} \StDerv\; \left\vert\;\; \odn{\;\StDerv\;}{}{\StDerv}{}\right.
\end{array}\]

The mapping $\cneg\cdot$ is extended to all standard formulae by defining $\cneg\cdot\colon\StFmla\to\StFmla$ inductively as
$$
\cneg{\cneg a}=a
\quad
\cneg\zer=\one
\quad
\cneg\one=\zer
\quad
\ccneg{A\vlan B}=\cneg A\vlor \cneg B
\quad
\ccneg{A\vlor B}=\cneg A\vlan \cneg B
$$
\end{definition}

We inherit all the applicable definitions from Section \ref{sec:prelims}, noting that in standard systems atoms behave like units, and not like connectives, and we have neither explicit nor guarded substitutions and therefore no composition by expansion.

We define \defn{standard derivations} to be correct standard pre-derivations according to the inference rules of $\SKS$~\cite{brunnler:tiu:01,Brun:04:Deep-Inf:rq,GuglGundPari::A-Proof-:fk},\footnote{For convenience, we give the inference rules of $\SKS$ in the appendix.}  which is sound and complete for classical propositional logic.\looseness=-1

\begin{definition}\label{def:synch}
  Let $\Phi$ and $\Psi$ be two standard derivations such that $\cn\Phi\ideq\pr\Psi$. A \defn{synchronal composition} of $\Phi$ and $\Psi$ is a standard derivation, denoted $\downsmash{\vls\odt\Phi{}\Psi{}}$, that can inductively obtained as follows: %

  \noindent
  \begin{tabular}{@{}l@{\;}l@{\quad}l@{\;}l}
    if $\Phi \in \StFmla$,&
    then 
    $\vls\odt\Phi{}\Psi{}\ideq\Psi$; &
    if $\Psi \in \StFmla$, &
    then 
    $\vls\odt\Phi{}\Psi{}\ideq\Phi$;\\
    if $\Phi\ideq\odn{\Phi_1}{}{\Phi_2}{}$, &
    then 
    $\od{\odt\Phi{}\Psi{}};
    \ideq
    \od{\odi{\odh{\Phi_1}}{}{\odt{\Phi_2}{}\Psi{}}{}}$; &
    if $\Psi\ideq\odn{\Psi_1}{}{\Psi_2}{}$, &then 
    $\od{\odt\Phi{}\Psi{}}
    \ideq
    \od{\odi{\odh{\odt\Phi{}{\Psi_1}{}}}{}{\Psi_2}{}}$;
  \end{tabular}

if $\begin{cases}
  \vls\Phi\ideq(\Phi_1\alp\Phi_2),\\
  \text{and }
\vls\Psi\ideq(\Psi_1\alp\Psi_2), \\
\text{with } \alp \in \Conn
\end{cases}$ then
$\vls
\odt{\Phi}{}{\Psi}{}
\ideq\vls{\odt{\Phi_1}{}{\Psi_1}{}\alp\odt{\Phi_2}{}{\Psi_2}{}}$\;.
\vadjust{\vskip6pt}
\end{definition}
Note that this operation is not necessarily uniquely defined. But whenever we use this notation, any synchronal composition of $\dera$ and $\derb$ can be used.

\begin{definition}\label{defn:intof}
  We define the (partial) \defn{interpretation map} $\intof{\,\cdot\,}\colon \SADerv \to \StDerv$ from subatomic derivations to standard derivations as indicated in Figure~\ref{fig:interpretation}. On the top of that figure we show the case analysis for the base case $\intof{\Phi\ba\Psi}$.
\end{definition}
\begin{figure}
  \def\pif{\phantom{if~}}
  \begin{align*}
  \intof {\Phi \ba \Psi} :~&
\mbox{\small%
    \begin{tabular}{|c|c|c|c|}   
    \hline 
    & $\intof\Psi = \one$ & $\intof\Psi = \zer$ & \makecell{$\intof{\pr\Psi} = \zer$ \\ $\intof{\cn\Psi} = \one$} \\
         \hline
       $\intof\Phi = \one$  & $\one$ & $\cneg{a}$ & \odn{\cneg{a}}{\awU}{\one}{} \\  \hline
       $\intof\Phi = \zer$  & $a$ & $\zer$ & \odn{\zer}{\awD}{a}{} \\  \hline
        \makecell{$\intof{\pr\Phi} = \zer$ \\ $\intof{\cn\Phi} = \one$} & \odn{a}{\awU}{\one}{} & \odn{\zer}{\awD}{\cneg a}{} & \odn{\zer}{\mix}{\one}{} 
        \\  \hline
    \end{tabular}
} \\ \\
    \intof{\one} &= \one, \\
    \intof{\zer} &= \zer, 
    \\
  \intof{\Phi \vlan \Psi} &=
  \left\{ \begin{array}{@{\,}ll}
    \mbox{undefined}&\mbox{if one of $\intof\Phi$, $\intof\Psi$ is undefined}\\
    \intof{\Phi} & \mbox{if~} \intof\Psi =\one\\
    \intof{\Psi} & \mbox{if~} \intof\Phi =\one\\
    \zer &  \mbox{if~}\intof\Phi =\zer \mbox{~and~} \intof\Psi =\zer \\
    \intof{\Phi} \vlan \intof{\Psi} &\mbox{otherwise}  
  \end{array}\right.
  \\
  \intof{\Phi \vlor \Psi} &=
  \left\{ \begin{array}{@{\,}ll}
    \mbox{undefined}&\mbox{if one of $\intof\Phi$, $\intof\Psi$ is undefined}\\
    \intof{\Phi} & \mbox{if~} \intof\Psi =\zer\\
    \intof{\Psi} & \mbox{if~} \intof\Phi =\zer\\
    \one &  \mbox{if~}\intof\Phi =\one \mbox{~and~} \intof\Psi =\one \\
    \intof{\Phi} \vlor \intof{\Psi} &\mbox{otherwise}  
  \end{array}\right.
  \\
  \intof{\esbs\Phi x \Psi} &=
  \left\{ \begin{array}{@{\,}ll}
     \intof{\asbs \Phi x \Psi} &\mbox{if~$\intof\Phi$ and $\intof\Psi$ are both defined} \\
     \mbox{undefined}&\mbox{otherwise}
  \end{array}\right.
  \\
  \intof{\gsbs\Phi x p\Psi} &=
  \left\{ \begin{array}{@{\,}ll}
     \intof{\agsbs \Phi x p \Psi} &\mbox{if~$\intof\Phi$ and $\intof\Psi$ are both defined} \\
     \mbox{undefined}&\mbox{otherwise}
  \end{array}\right.
  \\
  \intof{\odn\Phi{}\Psi{}} &=
  \left\{ \begin{array}{@{\,}ll}
    \downsmash{\odn{\intof{\Phi}}{}{\intof{\Psi}}{}} &\mbox{if~$\intof\Phi$ and $\intof\Psi$ are both defined} \\
    & \mbox{\pif and $\intof{\cn \Phi}\neq\intof{\pr \Psi}$}\\
     \downsmash{\odt{\intof{\Phi}}{}{\intof{\Psi}}{}} &\mbox{if~$\intof\Phi$ and $\intof\Psi$ are both defined} \\
    & \mbox{\pif and $\intof{\cn \Phi}=\intof{\pr \Psi}$}\\
         \mbox{undefined}&\mbox{otherwise}
  \end{array}\right.
  \\
  \intof{\odq\Phi{}\Psi{}} &=
  \left\{ \begin{array}{@{\,}ll}
     \downsmash{\odt{\intof{\Phi}}{}{\intof{\Psi}}{}} &\mbox{if~$\intof\Phi$ and $\intof\Psi$ are both defined} \\
    & \mbox{\pif and $\intof{\cn \Phi}=\intof{\pr \Psi}$}\\
         \mbox{undefined}&\mbox{otherwise}
  \end{array}\right.
\end{align*}
  \caption{The interpretation map $\intof{\,\cdot\,}\colon \SADerv \to \StDerv$}
 \label{fig:interpretation}
\end{figure}

\begin{proposition}[restate = PropDeri, name = ]
If $\dera$ is a subatomic derivation in $\TBLS$ and $\intof\dera$ is defined, then $\intof\dera$ is a standard derivation in~$\SKS$.  
\end{proposition}

\begin{example}
  Let us look again at the two derivations in Example \ref{ex:prederi}. Both $\intof\dera$ and $\intof\derb$ are undefined. However, we have\padjust{-2ex}
  $$
  \intof{\esbs{0}{w}\Psi} = \odn{(a.\odn{0}{\awD}{-a}{})}{\aiU}{0}{}
  \quand\!
  \intof{\esbs{1}{w}\Psi} = \odn{(a.\odn{-a}{\awU}{1}{})}{\fequiv}{a}{}
  $$
\end{example}

\begin{remark}
  It should be noted that the size of $\intof{\Phi}$ is not polynomially bounded by $\Phi$ since all explicit and guarded substitutions are converted to actual substitutions.
  Nevertheless in some situations we need not perform every substitution in order to determine the interpretation, so that it may be computed in polynomial time. An example is given by the Proposition below.
  \end{remark}

\begin{proposition}[restate = IntPoly, name = ]
    \label{prop:intpoly}
Let $A$ be an open formula with free variables $\set{v_1,\dots,v_n}$ and let $B_1,\dots,B_n$ be formulae. If $\intof{B_1}=\dots=\intof{B_n}=U\in\Unts$ then $\intof{\exs{B_i}{v_i}{A}A}=U$. 
\end{proposition}

\subsection{Soundness and Completeness}

The main purpose of the interpretation map is to carry the soundness and completeness of $\SKS$ to $\TBLS$.

\begin{definition}
  A \defn{(subatomic) proof} is a derivation $\dera$ such that $\intof{\pr\dera}=\one$. A subatomic formula $A$ is \defn{provable} in $\TBLS$ if there is a proof $\dera$ in $\TBLS$ with $\cn\dera=A$, and a standard formula $A$ is \defn{provable} in $\TBLS$ if there is a proof $\dera$ in $\TBLS$ with $\intof{\cn\dera}=A$.
\end{definition}

\begin{theorem}[restate = SoundCompl, name = ]
  \label{thm:sub-compl}
  Let $A$ be a standard formula. The following are equivalent:
  \begin{enumerate}
  \item $A$ is a Boolean tautology.
  \item $A$ is provable in $\SKS$.
  \item $A$ is provable in $\TBLS$.
  \item $A$ is provable in $\TBLS$ without the algebraic rules.
  \end{enumerate}
\end{theorem}

\begin{remark}
  The equivalence of 3) and 4) in Theorem~\ref{thm:sub-compl} indicates that the algebraic rules of~\eqref{eq:extra-rules} are, in fact, not needed. However, the presentation of the technicalities in the following sections greatly benefits from the presence of the algebraic rules and therefore they are included in $\TBLS$.
\end{remark}

\section{Cut Freeness and Superposition of Proofs}
\label{sec:cut}

\subsection{Cuts and projections in subatomic systems}

Figure~\ref{fig:cut} in the introduction shows that the atomic cut rule $\aiU$ is the interpretation of an instance of the $\sUp\ba\vlan$. This motivates the following definition.

\begin{definition}
A \defn{(subatomic) cut} is an instance of the rule scheme $\sUp\ba\vlan$ or $\vlan\sUp\ba$ shown below:
\begin{equation}
  \label{eq:cut-scheme}
          \odn{((x\ba              y)\vlan(z\ba w))}
{\sUp\ba\vlan}{\vlsmallbrackets
  ((x\vlan z)\ba(y\vlan w))}{}
\end{equation}
A subatomic proof system is \defn{cut-free} if it does not contain any cuts. 
We denote by $\TBLScf$ the cut-free fragment of $\TBLS$, that is, $\TBLS$ with the rules $\sUp\ba\vlan$ and $\vlan\sUp\ba$ removed for all $a\in\Atms$. 
\end{definition}

This definition is justified by the following proposition which says that any interpretation of a proof in $\TBLScf$ is also cut-free in $\SKS$.

\begin{proposition}
  If $\dera$ is a cut-free subatomic derivation then $\intof{\dera}$ contains no instances of the atomic cut rule.
  \rm\cite{BarrGugl:22:A-Subato:yr}
\end{proposition}

In \cite{BarrGugl:22:A-Subato:yr} it has been shown that any subatomic proof can be transformed into a cut-free subatomic proof, and in \cite{BarrGuglRalph:25:strlin} this result has been extended to the strictly linear proof system without unit equality inference steps.

In both cases, cut elimination is achieved via \emph{projections}. The \defn{left-projection} (resp. \defn{right-projection}) of a derivation $\dera$ with respect to an atom $\ba$, denoted by $\lpr\ba\dera$ (resp. $\rpr\ba\dera$) is obtained from $\dera$ by discarding everything in the right (resp. left) scope of an atom and collapsing inference rules whenever possible. 

\begin{example}\label{example:projection}
   The left- and right-projections of \[\dera\ideq\odn
  {(\odn
  {[(\zer\ba\one).(\zer\ba\one)]}
{\ba\sDn\vlor}{([\zer.\zer])\ba([\one.\one])}{}.\odbox{\odn
  {(\one.\zer)}
{\mix}{[\one.\zer]}{}\ba\odn
  {(\zer.\zer)}
{\mix}{[\zer.\zer]}{}})}
  {\vlan\sUp\ba}{(([\zer.\zer].[\one.\zer]))\ba(([\one.\one].[\zer.\zer]))}{}\]
  with respect to $\ba$ are
\[\lpr\ba\dera=
(\zer\vlor\zer)\vlan\odn{\one\vlan\zer}{\mix}{\one\vlor\zer}{}\quad\quad\rpr\ba\dera=
(\one\vlor\one)\vlan\odn{\zer\vlan\zer}{\mix}{\zer\vlor\zer}{}
\]
\end{example}

A subatomic derivation containing an atomic connective $\ba$ can be viewed as a juxtaposition of two derivations: the one in which that atom is false and the one in which it is true.

We can also read atomic connectives as decision trees~\cite{BarrGugl:22:A-Subato:yr}, so that $A\ba B$ means ``if $a$ is true then $B$ else $A$''. Correspondingly, the left-projection of a derivation with respect to $\ba$ is the derivation in which that atom is false, and its right-projection is the derivation in which it is true.

The left- and right-projections can then be recombined to obtain a formula which is equivalent under the interpretation map to the conclusion of the original proof:
\begin{equation}
  \label{eq:cut-deri}
\odv{((\lpr\ba\dera)) \ba ((\rpr\ba\dera))}{\derd}{C}{}
\end{equation}
where $\derd$ is a cut-free subatomic derivation. In a strictly linear proof system, where unit equality inference steps are not available, Lemma~\ref{lemma:mergemedial} can be applied, so any leaves $v$ or $u$ which are not in the scope of $\ba$ in $\cn \dera$ are replaced by $v\ba v$ or $u\ba u$ in $C$, which are nevertheless interpreted as $v$ or $u$. Examples of this are given in \cite{BarrGuglRalph:25:strlin}.

The construction of $\derd$ in~\eqref{eq:cut-deri} above, as detailed out in \cite{BarrGugl:22:A-Subato:yr} and~\cite{BarrGuglRalph:25:strlin}, can also be applied to our system. However, as in  \cite{BarrGugl:22:A-Subato:yr} and~\cite{BarrGuglRalph:25:strlin}, this construction would lead to an exponential blow-up in the size of the proof. For this reason, we will not repeat that construction here, as we will prove a much stronger result in Section~\ref{sec:psim} below.


\subsection{Superpositions}

One of the insights of our work is that if the translation of a standard derivation ($\SKS$ or Frege) into subatomic logic is done carefully, then the underlying structure of projections will match. The following example shows that this is not immediate.

\begin{example}\label{example:superposition}
The projections of $\dera$ given in Example~\ref{example:projection} have matching structures, in the sense that there exists a derivation 
\[\derb=
(x\vlor x)\vlan\odn{y\vlan\zer}{\mix}{y\vlor\zer}{}
\] such that $\asbss{\zer}{x}{\one}{y}\derb=\lpr\ba\dera$ and $\asbss{\one}{x}{\zer}{y}\derb=\rpr\ba\dera$.

The interpretation of $\dera$ is  \[
\intof{\dera} = \odn
  {(\odn
  {[a.a] }
{\acD}{a}{}.\odn
  {\zer}
{\awD}{-a}{})}
{\aiU}{\zer}{}
\] exhibiting a contraction, a weakening, and a cut.

However, the derivation \[\derc
\ideq\odn
  {(\odn
  {[(\zer\ba\one).(\zer\ba\one)]}
{\ba\sDn\vlor}{([\zer.\zer])\ba([\one.\one])}{}.
\odbox{\odn
  {(\one.\zer)}
{\mix}{[\one.\zer]}{}\ba \zer})}
{\vlan\sUp\ba}{(([\zer.\zer].[\one.\zer]))\ba(([\one.\one].\zer))}{}
\]
has the same interpretation as $\dera$, but its projections do not match:
\[
\lpr\ba\derc=(\zer\vlor\zer)\vlan\odn{\one\vlan\zer}{\mix}{\one\vlor\zer}{}\quad\quad\quad
\rpr\ba\derc=(\one\vlor\one)\vlan\zer
\]
\end{example}

\begin{definition}
We say that $\derb$ is a \defn{superposition} of $\dera$ with respect to an atom $\ba$ if there exist free variables $x$ and $y$ in $\derb$ such that $\asbss{\zer}{x}{\one}{y}\derb=\lpr\ba\dera$ and $\asbss{\one}{x}{\zer}{y}\derb=\rpr\ba\dera$.
\end{definition}
  
As said before, the method of cut elimination given in \cite{BarrGuglRalph:25:strlin} increases the size of a proof exponentially: to eliminate all of the cuts from a proof $\dera$, we project with respect to each atom in turn and merge the conclusions. Since the left- and right- projections can be almost as large as the original derivation, the elimination of each atom can double the size of the proof. However, by bringing together superpositions and guarded substitutions, we are able to compress this construction. 

\begin{example}\label{example:observation}
If $\derb$ is a superposition of $\dera$ with respect to $\ba$, then we can construct:
\[
\odq{\gsbss{\zer}{x}{\one}{y}{l} 
     \gsbss{\one}{x}{\zer}{y}{r}\odq
    {\esbs{\derb}{z} ((z\rng l\ba z\rng r)) }
{}{\range{\cn\derb}{l} \ba \range{\cn\derb}{r} }{} }
{}{ \asbss{\zer}{x}{\one}{y}\range{\cn\derb}{l} \ba \asbss{\one}{x}{\zer}{y}\range{\cn\derb}{r} }{}
\]
which has the same conclusion as $\lpr\ba\dera\ba\rpr\ba\dera$ but is in general much smaller because $O(\size \derb)= O\left(\size{\lpr\ba\dera}\right)=O\left(\size{\rpr\ba\dera}\right)$.
\end{example}

This idea forms the basis of how we are able to polynomially simulate substitution Frege systems with our cut-free system $\TBLScf$.

\section{Substitution Frege}
\label{sec:frege}

To substantiate our claim that guarded substitutions are a highly effective means of proof compression, we compare our cut-free system $\TBLScf$ to substitution Frege systems \cite{CookReck:79:The-Rela:mf}, which are a class of p-equivalent systems for propositional logic with a substitution rule. No system is known to be more powerful than substitution Frege systems in compressing the size of proofs, and this therefore serves as a benchmark. In this section, we show that any substitution Frege proof can be transformed into a $\TBLScf$ proof with only a polynomial increase in size.

\begin{definition}
  A \defn{Frege substitution} is a function $\rho\colon\Atms\to\StFmla$ which is the identity almost everywhere.
\end{definition}

We use for Frege substitutions the same notation as for actual substitutions, i.e., if $\set{a_1,\ldots,a_n}$ is the set of atoms on which $\rho$ is not the identity, and $\rho=\set{a_1\mapsto A_1,\ldots,a_n\mapsto A_n}$ then
\begin{equation}
  \label{eq:frege-sub}
  \rho B=\asb{\sbs{A_1}{a_1},\ldots,\sbs{A_n}{a_n}}B
\end{equation}
is the formula obtained from $B$ by simultaneously replacing all occurrences of $a_i$ with $A_i$ and all occurrences of $\cneg a_i$ with $\cneg A_i$ for all $i\in\set{1,\ldots,n}$.

\begin{definition} 
A \defn{Frege proof with substitution} $\Phi= P_1,\dots, P_h$ is a finite sequence of standard formulae such that for every $i\in\set{1,\ldots,h}$, the formula $P_i\in\StFmla$ is one of the following:
\begin{itemize}
\item An instance of one of the axioms below:
  \begin{equation*}
    \begin{array}{r@{\;\ideq\;}l}
      \frF_1& \bar{A} \vlor (B \vlor A)            \\         
      \frF_2&(A \vlan (B\vlan \bar C))\vlor ((A\vlor \bar B)\vlor (\bar A\vlor C))\\
      \frF_3&(\bar A\vlan B)\vlor (\bar B \vlor A)                   \\
       \frF_4&\one     \\
\end{array}
  \end{equation*}
\item The conclusion of an instance of the \defn{modus ponens rule}: 
  \begin{equation}
    \label{eq:mp}
  \vliinf\mp{}{B}A{\bar A\vlor B}  
  \end{equation}
  where the premises are $P_k$ and $P_l$ for some $k,l < i$.
\item The conclusion of the \defn{substitution rule}: 
  \begin{equation}
    \label{eq:sub}
  \odn B\sub{\rho B}{}  
  \end{equation}
  where the premise is $P_k$ for some $k<i$, and $\rho$ is an arbitrary Frege substitution.
  \end{itemize}
The conclusion of $\Phi$ is $P_h$, the size is given by $\sizeof\Phi=\sum{\sizeof{P_i}}$, and the width by $\max{\sizeof{P_i}}$. We denote by $\sF$ the proof system that comprises all Frege proofs with substitution.
\end{definition}

\begin{definition}
  We denote by $\sFzo$ the proof system that comprises all Frege proofs with substitution where for every instance of the substitution rule, the image of the substitution $\rho$ is $\Unts$, i.e., all formulae $A_1,\ldots,A_n$ in~\eqref{eq:frege-sub} are either $\zer$ or $\one$.
\end{definition}

In the case of $\sFzo$, we can write~\eqref{eq:frege-sub} as follows
$$\rho B=\asb{\sbs{A_1}{a_1},\ldots,\sbs{A_n}{a_n}}B=\asbs{A_1}{a_1} \ldots \asbs{A_n}{a_n}B$$

\begin{theorem}[restate = sFregezo, name = ]
  \label{thm:sfzo}
  $\sF$ and $\sFzo$ are p-equivalent. \rm\cite{Buss:95:Some-Rem:xe}
\end{theorem}

We are now ready to state the main theorem of this paper.

\begin{theorem}[restate = MainThm, name = ]
  \label{thm:main}
  System $\TBLScf$ p-simulates $\sF$.
\end{theorem}

In other words,  there is a $k\in \Nat$ such that for every $\sFrege$ proof  $\Phi= P_1,\dots, P_h$, there is a derivation $\derb$ in $\TBLScf$, such that $\intof{\pr\derb}=\one$ and $\intof{\cn\derb}=P_h$ and  $\sizeof{\derb}=O(\sizeof{\Phi}^k)$.

The following section is dedicated to the proof of Theorem~\ref{thm:main}. We will make use of Theorem~\ref{thm:sfzo}, and show p-simulation of $\sFzo$, as this simplifies the presentation.

\section{P-simulation of Substitution Frege}
\label{sec:psim}

The proof of Theorem~\ref{thm:main} will be in three phases.
In the first phase, we construct a $\TBLScf$ derivation $\derb$ which tracks the relations between formulae in a given $\sFzo$ proof $\dera$, using explicit substitutions and ranges.

This derivation is cut-free; in fact, it contains no atomic connectives at all. We can view this derivation as a superposition with respect to every atom which occurs in the $\sFzo$ proof~$\Phi$. The atoms in $\Phi$ correspond to \textit{variables} in $\derb$, in the same way that the atom $a$ corresponds to the variables $x$ and $y$ in Example \ref{example:superposition}.

The second and the third phase of the simulation take the derivation $\derb$ produced in the first phase and apply a rearrangement of subformulas (phase 2) and guarded substitutions (phase 3) to obtain a proof in $\TBLScf$ whose interpretation is a proof of the conclusion of the Frege proof $\dera$.

\subsection{Simulating \texorpdfstring{$\mp$}{mp} and \texorpdfstring{$\sub$}{sub} in \texorpdfstring{$\TBLScf$}{KSubG}}

In a $\sFzo$ proof, a formula can act as a premise for multiple applications of the $\mp$ and $\sub$ rules. 

\begin{example}\label{example:fregeproof}
Consider a $\sFzo$ proof $P_1,\dots,P_h$ which contains the following fragment: \padjust{-2ex}
\[
\begin{array}{rlc}
 & \vdots &   \\
 P_{k-2}&= ((\bar a\vlor a)\vlan( b\vlan\bar b))\vlor(c\vlor \bar c) &\\
 P_{k-1}&= \bar a \vlor a & \\
 P_k &=   (a\vlan \bar a)\vlor(\bar b\vlor b) &    \\
 P_{k+1}&= (\ttt\vlan \fff)\vlor(\bar b\vlor b) &  (\sub (P_k)) \\
 P_{k+2}&= \bar b\vlor b & (\mp(P_{k-1}, P_k)) \\
 P_{k+3}&= c\vlor \bar c & (\mp(P_k,P_{k-2} )  ) \\
 P_{k+4}&=   (a\vlan \bar a)\vlor(\fff\vlor \ttt) &  (\sub (P_k))  \\
   &\vdots& \\
\end{array}
\]
Then $P_k$ is used as a premise in obtaining $P_{k+1}$,  $P_{k+2}$, $P_{k+3}$, and $P_{k+4}$.
\end{example}

It has been shown in~\cite{BrusGugl:07:On-the-P:fk} that $\SKS$ p-simulates Frege systems, and that $\SKS$ extended by a substitution rule p-simulates $\sFrege$.\footnote{In fact the authors of \cite{BrusGugl:07:On-the-P:fk} use the variant $\SKSg$, obtained from $\SKS$ by loosening the restriction that the structural rules should be atomic, but it is very straightforward to show p-equivalence between $\SKS$ and $\SKSg$.} In order to simulate instances of $\mp$, it is assumed in~\cite{BrusGugl:07:On-the-P:fk} that the conjunction of all previous lines of the proof is proven. The premises for the $\mp$ inference steps are duplicated via cocontractions and used to prove the $\mp$ conclusion using switch and cut. In a strictly linear proof system, it is cumbersome to duplicate and permute formulae, but we can profit from using explicit substitutions instead.

\begin{example}\label{ex:frege-mp}
The applications of modus ponens in Example \ref{example:fregeproof} could be simulated by a construction like the following:
\[
\odq{
\esbs{P_k}{y_k}
\esbs{ P_{\vlnos k-1}\vlan y_k }{y_{k+2}}
\esbs{y_k \vlan P_{\vlnos k-2} }{y_{k+3}}\dots y_h }
{}{
\esbs{\odframefalse\odv
{P_{\vlnos k-1}\vlan P_k}
{\Psi_{k+2}}{P_{k+2}}{\Ssys}}{y_{k+2}}
\esbs{\odframefalse\odv
{P_k \vlan P_{\vlnos k-2}}
{\Psi_{k+3}}{P_{k+3}}{\Ssys}}{y_{k+3}}\dots y_h}{}
\] 
where $\Psi_{k+2}$ and $\Psi_{k+3}$ perform the modus ponens in some proof system~$\Ssys$, and each $y_i$ is a fresh variable. That is, each successive step of the $\sFzo$ proof is proven inside an explicit substitution, and then the formula is substituted every time the conclusion is used further down in the proof.

Note that this construction is not yet precise; in particular $\Psi_{k+2}$ and $\Psi_{k+3}$ are not strictly linear, and the formulae are not translated into subatomic logic.
\end{example}

In $\SKS$, atomic cuts, switches and unit equality rules are used to cancel out the antecedent $A$ of modus ponens~\eqref{eq:mp} and its negation $\cneg A$ in the other premise, and the resulting formula $B$ can then be used as a premise in future inference steps (just as in $\sF$) without duplicating any material unnecessarily. 

In our strictly linear system, unit equality inference steps are not available, and so we are unable to destroy the cut formulae. Instead of cuts, we can use explicit substitutions to ensure that when we substitute the conclusion into future inference steps, we do not bring along material which will blow up the size of the proof, using Lemma \ref{lemma:mergeswitch} to transform $A\vlan \bar A$ into a formula composed entirely of conjunctions of $\one$ and $\zer$, which will therefore vanish under the interpretation function.\looseness=-1

But observe that we would need the derivation on the right in Lemma~\ref{lemma:mergeswitch}, using the rules $\sUp\bet\vlan$ and $\vlan\sUp\bet$. If $\bet\in\Atms$ then this derivation will not be cut-free. For this reason, we have to ensure that our derivation $\derb$ does not contain any atomic connectives, and so we use instead a superposition. The atomic connectives are recovered in a later phase in the construction.

Besides having to simulate modus ponens, we also need a way of dealing with substitution inference steps. Not only are we constrained by the linearity of our proof system, but also by how we can apply substitutions. In $\sFzo$, substitutions arrive \emph{ex nihilo}: we can proceed from $C$ to $\asbs{B}{a}C$, for any formula $B$. This is not allowed in $\TBLS$, where substitutions only move between being explicit and actual: we must proceed from $\esbs{B}{a}C$ to $\asbs{B}{a}C$.

Nevertheless the finer control over where substitutions apply afforded by guarded substitutions allows us to simulate the substitutions of $\sFzo$.

\begin{example}\label{ex:frege-sub}
In Example \ref{example:fregeproof}, the substitution rule is applied twice to $P_k$, to obtain $P_{k+1}$ and $P_{k+4}$. In order to achieve the same proof compression as $\sFzo$, we want to prove $P_k$ once and reuse it, without any conflict between the substitutions.

Limiting ourselves to explicit substitutions, and using the idea from Example~\ref{ex:frege-mp} we might consider derivations with either of the premises
\[
\esbs{\ttt}{x_b}\esbs{\fff}{\varneg x b}\esbs{\ttt}{x_a}\esbs{\fff}{\varneg x a} \esbs{P_k}{y_k} \esbs{y_k}{y_{k+1}}\esbs{y_k}{y_{k+4}}\dots y_h
\]
\[
\esbs{P_k}{y_k} \esbs{\esbs{\ttt}{x_a}\esbs{\fff}{\varneg x a}y_k}{y_{k+1}} \esbs{\esbs{\ttt}{x_b}\esbs{\fff}{\varneg x b}y_k}{y_{k+4}} \dots y_h
\]
where the variables $x_a$ and $\varneg x a$ represent the atom $a$ (see Example~\ref{example:superposition}), and the variable $x_b$ and $\varneg x b$ represent the atom $b$, and the explicit substitutions $\esbs{\ttt}{x_b}$, $\esbs{\fff}{\varneg x b}$, $\esbs{\ttt}{x_a}$, and $\esbs{\fff}{\varneg x a}$ are intended to simulate the substitution inference steps.
In the first case, the substitutions $\esbs{\ttt}{x_a}$ and $\esbs{\fff}{\varneg x a}$ will also catch $P_{k+4}$; in the second case, we would need to apply $\esbs{P_k}{y_k} $ in a variable capturing way, which we cannot do.
However, by introducing guarded substitutions, we can use the premise
\[\gsbss{\ttt}{{x_b}}{\fff}{\varneg x b}{s}\gsbss{\ttt}{x_a}{\fff}{\varneg x a}{r}\esbs{P_k}{y_k}\esbs{y_k\rng{r}}{y_{k+1}}\esbs{y_k\rng {s}}{y_{k+4}}\dots y_h \]
which can be expanded as
\[
\odq
{\gsbss{\ttt}{{x_b}}{\fff}{\varneg x b}{s}\odq
{\gsbss{\ttt}{x_a}{\fff}{\varneg x a}{r}\odq
{\esbs{P_k}{y_k}\esbs{y_k\rng{r}}{y_{k+1}}\esbs{y_k\rng {s}}{y_{k+4}}\dots y_h }
{}{\esbs{\range{P_k}{r}}{y_{k+1}}\esbs{\range{P_k}{s}}{y_{k+4}}\dots y_h}{}}
{}{\esbs{\asbss{\ttt}{x_a}{\fff}{\varneg x a}
         \range{P_k}{r} }{y_{k+1}}
   \esbs{\range{P_k}{s} }{y_{k+4}}
   \dots y_h }{} }
{}{\esbs{\asbss{\ttt}{x_a}{\fff}{\varneg x a}
         \range{P_k}{r} }{y_{k+1}}
   \esbs{\asbss{\ttt}{{x_b}}{\fff}{\varneg x b}
         \range{P_k}{s} }{y_{k+4}}
   \dots y_h }{}
\]
in order to obtain the formulae $P_{k+1}$ and $P_{k+4}$. We note that this construction accumulates the ranges $r$ and $s$, so that the size of these formulae grow, but this will not have a greater than polynomial effect on the size of the construction. 
\end{example}

\subsection{From Standard Formulae to Subatomic Formulae}

Every line in a $\sFzo$ proof is a standard formula. We now show how we translate such a standard formula into a subatomic formula.

\begin{definition}
An actual substitution is \defn{elementary} if it maps variables only to units or variables. An actual substitution $\sigma$ \defn{closes} a formula $A$ if $\sigma A$ is closed. An actual substitution $\sigma$ is \defn{dualising} with respect to two variables $x$ and $y$ if $\sigma x=\ccneg{\sigma y}$.
\end{definition}
\begin{definition}
\label{defn:vblenumeration} 
We call any finite sequence of positive and negative atoms $\atmenum{n}=\set{a_1,\atmneg{a}{1}\dots,a_n,\atmneg{a}{n}}$ an \defn{atomic enumeration}. If $\Phi$ is an $\sFzo$ proof, we write $\atmenum{\Phi}$ for the set of atoms appearing in $\Phi$, including the pair of dual atoms, even if only one of the pair appears in the proof. Given an injective map from $v\colon\Atms \to \Vbls$, we define a \defn{variable enumeration}, $\varenum{\atmenum{n}}$ to be the image of an atomic enumeration $\atmenum{n}$, writing $\varenum{n}$ for $\varenum{\atmenum{n}}$ and $\varenum{\Phi}$ for $\varenum{\atmenum{\Phi}}$. We call the image of dual atoms $a_j$ and $\atmneg{a}{j}$ in a variable enumeration $\varenum{n}$ \defn{pseudo-dual variables}, writing them as $v(a_j)=x_j$ and $v(\atmneg{a}{j})=\varneg{x}{j}$, but we emphasise that this is not an extension of the negation operation $\atmneg{\cdot}{}$ to variables. Instead, for a subatomic formula $A$, given a variable enumeration $\varenum{n}$ we can define the \defn{pseudo-dual formula} $\Tilde{A}$ to be $\asbs{\varneg{x}{i}}{x_i}_{\fv{A}}\cneg{A}$.
\end{definition}

\begin{definition}
Given an atomic enumeration $\atmenum{n}$ with variable enumeration $\varenum{n}$, and a subatomic formula $A$ with $\fv A \subseteq \varenum{n}$, we say that $A$ is \defn{tautological} with respect to $\varenum{n}$ if for every elementary substitution $\eta$ which is dualising with respect to all pseudo-dual variables, we have $\intof{\eta A}=\one$, and it is \defn{contradictory} with respect to $\varenum{n}$ if $\intof{\eta A}=\zer$.
\end{definition}

\begin{definition}\label{defn:factorisation}
Given a formula $Q$, we define its \defn{factorisation} to be composed of two parts: an open formula $A$ whose every variable appears only once and an elementary substitution $\eta$, such that $\eta A=Q$. We note that the factorisation is unique modulo the choice of free variables of $A$, and assume in general that these are chosen to be fresh variables.
\end{definition}
We now give the translation that we will use to translate formulae of an $\sFzo$ proof to the subatomic system.

\begin{definition}
Let $\atmenum{n}$ be an atomic enumeration with variable enumeration $\varenum{n}$, and let $P$ be a standard formula whose every atom is in $\atmenum{n}$. We obtain the \defn{open translation} $Q$ of $P$ by replacing each atom by its image under the given variable enumeration, so that $\fv{Q}\subseteq \varenum{n}$. We can then factorise $Q$ to obtain an open formula $A$ and an elementary substitution $\eta$ such that $\eta A=Q$; we call the pair $(\eta,A)$ the \defn{factorised (open) translation} of $P$. For any two standard formulas $P_i,P_j$, we assume without loss of generality that in their factorised translations $\fv{A_i}\cap \fv{A_j}=\emptyset$.
\end{definition}

It is important to note that the factorised open translation does not contain any atomic connectives, but that we have replaced each atom by an associated variable. This guarantees that we will not use the cut rule in the first phase of the simulation, because there are no atoms to cut.

\begin{proposition}[restate = StToSub, name = ]
Let $\atmenum{n}$ be an atomic enumeration with variable enumeration $\varenum{n}=\set{x_1,\varneg x 1,\dots,x_n,\varneg x n}$ and let $P$ be a standard formula with factorised translation $(\eta,A)$. Then $\intof{\asbss{\zer \ba_j \one}{x_j}{\one\ba_j\zer}{\varneg x j}_{1\leq j\leq n} \eta A }=P$.
\end{proposition}


\subsection{Phase I: Superposition}
\begin{figure*}[!t]
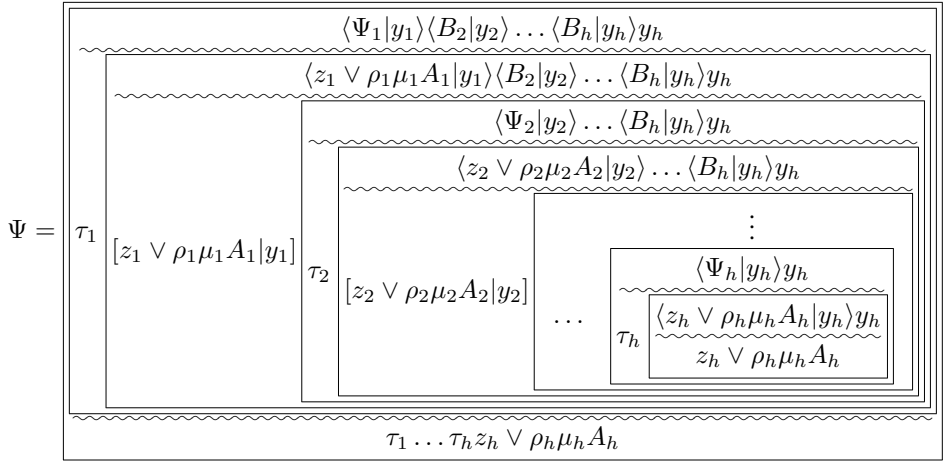

\[\derb\ideq
\od{\odq{\odq
  {\esbs{\derb_1}{y_1}\esbs{B_2}{y_2}\dots\esbs{B_h}{y_h}y_h}
{}{\tau_1\; \odq
  {\esbs{z_1 \vlor\rho_1\mu_1 A_1}{y_1}\esbs{B_2}{y_2}\dots\esbs{B_h}{y_h}y_h}
{}{\asbs{z_1 \vlor\rho_1\mu_1 A_1}{y_1}\odq
  {\esbs{\derb_2}{y_2}\dots\esbs{B_h}{y_h}y_h}
{}{\tau_2\; \odq
  {\esbs{z_2 \vlor\rho_2\mu_2 A_2}{y_2}\dots\esbs{B_h}{y_h}y_h}
{}{\asbs{z_2 \vlor\rho_2\mu_2 A_2}{y_2} \odbox{\begin{array}{cc}
     & \vdots \\
    \dots & \odq
  {\esbs{\derb_h}{y_h}y_h }
{}{\tau_h\; \odq
  {\esbs{z_h \vlor\rho_h\mu_h A_h }{y_h}y_h }
{}{z_h \vlor\rho_h\mu_h A_h }{} }{}
\end{array} } }{} }{}} {} }{}}
{}{\tau_1\dots\tau_h z_h \vlor \rho_h\mu_h A_h}{}}
\]
\caption{The derivation $\derb$ in Phase I of the p-simulation construction}
\label{figure:phase1}
\end{figure*}
\begin{convention}
In all the following we will use:
  \begin{itemize}
      \item $\sigma, \nu$ to denote guarded substitutions $\gsbs A y p$
      \item $\tau$, to denote explicit substitutions $\esbs A y$
      \item $\rho,\eta,\mu$ to denote actual substitutions $\asbs A y$
  \end{itemize} 
\end{convention}
\begin{construction}
\label{cons:phase-1-a}
Let $\dera=P_1,\dots,P_h$ be an $\sFzo$ proof, with atomic enumeration $\atmenum{\dera}$, variable enumeration $\varenum{\dera}$, and factorised open translation $(\eta_1,A_1),\dots,(\eta_h,A_h)$.

For each line $P_i$ of the $\sFzo$ proof $\dera$, we will construct:
\begin{itemize}
    \item a cut-free derivation $\derb_i$, 
    \item explicit substitutions $\tau_i$,
    \item actual substitutions $\mu_i$ and $\rho_i$
    \item a guarded substitution $\sigma_i$
    \item a flat formula $B_i$
\end{itemize}

The premise of each derivation $\derb_i$ will be determined by the premise(s) of the inference rule by which $P_i$ is obtained; when $P_i$ is an axiom, $\pr{\derb_i}$ will be $B_i$ and otherwise the formula $B_i$ will allow us to substitute in the premise needed. The guarded substitution $\sigma_i$, which is non-empty only when $P_i$ is obtained by an application of substitution, will store the information of that substitution until a later phase in the simulation. The conclusion of $\derb_i$ will be $\tau_i z_i \vlor \rho_i\mu_i A_i$, where $z_i$ is a fresh variable. The explicit substitution $\tau_i$ allows us to control the accumulation of logical material from applications of modus ponens, the actual substitution $\rho_i$ allows us to accumulate ranges so that the guarded substitutions apply correctly, and the actual substitution $\mu_i$ aids in giving strictly linear simulations of the axioms which make use of weakening and contraction. We fix a set of fresh variables $\set{y_1,\dots,y_h}$, which we will use in Construction \ref{cons:phase-1-b} to link these pieces together.

First suppose that $P_i$ is obtained from an axiom. We only show here the case for $\frF_1$, the others being similar (and shown in the appendix). We have that $A_i$ is of the form $\vls{[X'.(Y.X)]}$, where $\eta_i X'$ and $\eta_i X$ are pseudo-dual formulae. For a variable $v\in \fv X$, we denote by $v'$ the variable in $\fv X'$ appearing in the corresponding position, and we note that $\eta_i v$ and $\eta_i v'$ are pseudo-dual. With use of Lemma \ref{lemma:mergeswitch} we construct:\padjust{-2ex}
\[\Psi_i = \od{\odq{\zer}{}{\esbs{\zer}{z_i}z_i}{} \vlor \eta_i \odv
  {[\acs{\odframefalse
  \odn{(v.\zer)}\mix{[v.\zer]}{}}{v}{Y}Y.
  \odv{\acs{v'\vlor v}{v}{X}\sUp X}{}{[X'.X]}{\sDn\vlan\vlor,\vlor\sDn\vlan}]}
{}{[X'.[\acs{v\vlor 0}{v}{Y}Y.X]]}{\assD,\comD}}\]
and we obtain the substitutions $\tau_i=\esbs{\zer}{z_i}$, $\rho_i=\eta_i$, $\mu_i=\acs{v\vlor \zer}{v}{Y}$, $\sigma_i=\epsilon$, and the formula $B_i=\zer\vlor\eta_i(\acs{v\vlan \zer}{v}{Y}Y \vlor \acs{v'\vlor v}{v}{S}\sUp X ) $.

  Next suppose that $P_i$ is obtained from an application of modus ponens, so that there exist $k,l<i$ such that $P_l=\bar P_k \vlor P_i$. Then $A_l$ is of the form $A_{k,l}\vlor A_{i,l}$, where $\eta_l A_{k,l}$ and $\eta_k A_k$ are pseudo-dual formulae and $\eta_l A_{i,l}=\eta_i A_i$. For a variable $v$ occurring in $A_k$, we denote by $v'$ the variable occurring in $A_{k,l}$ in the corresponding position, and we note that $\eta_k v$ and $\eta_l v'$ are pseudo-dual. Then by Lemma \ref{lemma:mergeswitch}, we construct:\padjust{-2ex}
\[\Psi_i\ideq \od{\odi{\odi{\odh
{([z_k. \rho_k\mu_k A_k].[z_l.[\rho_l\mu_l A_{k,l}.\rho_l\mu_l A_{i,l}]])}}
{\sw}{\odbox{[[z_k.z_l]. \odn
  {(\rho_k\mu_k A_k.[\rho_l\mu_l A_{k,l}.\rho_l\mu_l A_{i,l}])}
{\sw}{[ \odframetrue\odv
  {( \rho_k\mu_k A_k. \rho_l\mu_l A_{k,l})}
{}{\acs{ \rho_k\mu_k v \vlan\rho_{l}\mu_l v' }{v}{A_k}\sDn {A}_k}{\set{\vlan\sUp\vlor, \sUp\vlor\vlan} }
. \rho_l\mu_l A_{i,l}]}{} ]}}{}}
{\assD}{ \odq{(([z_k.z_l]\vlor 
   \acs{ \rho_k\mu_k v \vlan \rho_{l}\mu_l v' }{v}{A_k}\sDn {A}_k))}{}{\esbs{([z_k.z_l])\vlor 
   \acs{ \rho_k\mu_k v \vlan \rho_{l}\mu_l v' }{v}{A_k}\sDn {A}_k }{z_i} z_i}{} \vlor  \rho_l\mu_l A_{i,l} }{}
}
\]
The substitution $\tau_i$ is $\esbs{(z_k\vlor z_l) \vlor \acs{\rho_k\mu_k v \vlan \rho_{l}\mu_l v }{v}{A_k}\sDn {A}_k}{z_i}$, $B_i=y_k\vlan y_l$, and $\sigma_i=\epsilon$. We inherit $\mu_i$ and $\rho_i$ from $\mu_l$ and $\rho_l$: $\mu_i A_i= \mu_l A_{i,l} $ and $\rho_i=\rho_l$.

  Finally suppose that $P_i$ is obtained from an application of substitution, so that there exists $k<i$, atoms $a_{j_1},\dots, a_{j_m}$ and units $u_{j_1},\dots, u_{j_m} $ such that $P_i =\asbs{u_{j_1}}{a_{j_1}}\dots\asbs{u_{j_m}}{a_{j_m}}P_k$. For a variable $v$ occurring in $A_i$, we denote by $v'$ the variable occurring in $A_k$ in the corresponding position. We construct
\[
\Psi_i\ideq \od{\odq{z_k\rng i}{}{\esbs{z_k\rng i}{z_i}z_i}{}\vlor \range{\rho_k A_k}{i} }
\]
\[\sigma_i=\gsbss{u_{j_1}}{x_{j_1}}{\bar u_{j_1}}{\varneg{x}{j_1}}{i}\dots\gsbss{u_{j_m}}{x_{j_m}}{\bar u_{j_m}}{\varneg{x}{j_m}}{i}\]
Then $B_i=y_k\rng i$ and $\tau_i = \esbs{z_k\rng i}{z_i}$ and $\mu_i$ is inherited from $\mu_k$: $\mu_i A_i=\mu_k A_k$. The substitution $\rho_i$ adds the range $i$ to every variable in its scope: $\rho_i = \acs{\rho_k v\rng{i}}{v}{\mu_iA_i}$. We note that $\size{\rho_i\mu_iA_i}\leq \size{\rho_k\mu_k A_k }+\size{A_k}$.
\end{construction}

\begin{lemma}[restate = PhaseOneTautology, name = ]
    \label{lemma:tautology}
For each $i\in\set{1,\dots,h}$, $\asbs{B_1}{y_1}\dots \asbs{B_i}{y_i}y_i$ is tautological with respect to $\varenum{\dera}$.
\end{lemma}

\begin{lemma}[restate = PhaseOneContradiction, name = ]
    \label{lemma:contradictory}
For each $i\in\set{1,\dots,h}$, $\tau_1\dots \tau_i z_i$ is contradictory with respect to $\varenum{\dera}$.
\end{lemma}

The following lemma shows that the ranges and guarded substitutions that we introduce interact in the intended way.

\begin{lemma}[restate = SigmaRhoEta, name = ]
    \label{lemma:sigmarhoeta}
For each $i$ and each variable $v\in\fv{A_i}$,  there exists a range $R\subseteq \set{1,\dots,i}$ and a derivation \[\odv{\sigma_h\dots\sigma_1\rho_i\mu_i v}
{}{\eta_i\mu_i v^R}{}\] whose width is $O(hw)$, whose only vertical composition is by expansion and whose height is $O(h)$.
\end{lemma}

\begin{lemma}[restate = PsiIPoly, name = ]
  \label{lemma:derbipoly}
For each $i$, the size of $\derb_i$ is polynomial in the size of $\dera$.
\end{lemma}

\begin{construction}
\label{cons:phase-1-b}
Let $\dera=P_1,\dots,P_h$ be an $\sFzo$ proof, with atomic enumeration $\atmenum{\dera}$, variable enumeration $\varenum{\dera}$, and factorised open translation $(\eta_1,A_1),\dots,(\eta_h,A_h)$. Using the derivations, substitutions, and formulae from Construction \ref{cons:phase-1-a}, we construct the $\TBLScf$ derivation $\derb$ given in Figure \ref{figure:phase1}.
\end{construction}

\begin{lemma}[restate = PsiPoly, name = ]
  \label{lemma:phase1}
The size of $\derb$ is polynomial in the size of $\dera$: its width is $O(h^2w)$ and its height is $O(hw)$.
\end{lemma}

\subsection{Phase II: Eversion}

The derivation $\derb$ given in Figure \ref{figure:phase1} is designed to be a superposition of the derivation we need to work on. We now generalise the approach described in Example \ref{example:observation}. For this we have to bring the correct atoms back into the conclusion of $\derb$. This is done by Lemma~\ref{lemma:phase2eversion} below, which is an example of what is called \emph{eversion} in \cite{BarrGuglRalph:25:strlin}. 

\begin{definition}
Given an atomic enumeration $\atmenum{n}$ and any variable $v$, we define the formula $\expr{n}{v}$ to be 
\[\expr{n}{v}=
\esbs{v\rng{l_1}\ba_1 v\rng{r_1}}{v}\dots\esbs{v\rng{l_n}\ba_n v\rng{r_n}}{v}v
\] for $l_1,r_1,\dots,l_n,r_n$ guards. We define the formula $\expri{n}{j}{v}$ to be the formula obtained from $\expr{n}{v}$ by omitting the $j^{th}$ term; that is, the formula associated with the atomic enumeration $\atmenum{n}\setminus \set{a_j,\atmneg{a}{j}}$. 
Given a variable enumeration $x_1, \varneg{x}{1}\dots,x_n, \varneg{x}{n}$ of $\atmenum{n}$, for each $j$, we define the guarded substitution \[\nu_j=\gsbss{\zer}{x_j}{\one}{\varneg{x}{j}}{l_j}\gsbss{\one}{x_j}{\zer}{\varneg{x}{j}}{r_j}\] 
\end{definition}

\begin{proposition}[restate = Expprestaut, name = ]
    \label{prop:expprestaut}
Let $\atmenum{n}$ be an atomic enumeration and let $\varenum{n}$ be a variable enumeration. Let $A$ be a formula such that $\fv A \subseteq \varenum{n}$. If $A$ is a tautology (resp. contradiction) with respect to $\varenum{n}$ then $\intof{\nu_n\dots\nu_1 \esbs{A}{y} \expr{n}{y}}=\one$ (resp. $\zer$).
\end{proposition}

\begin{lemma}[restate = PhaseTwoEversion, name = ]
    \label{lemma:phase2eversion}
Let $\atmenum{n}$ be an atomic enumeration, let $y$ be a variable and let $A$ be a flat open formula. Then there exists a $\TBLScf$ derivation $\odv{\esbs{A}{y}\expr{n}{y} }{}{\acs{\expr{n}{v}}{v}{A}A }{\alp\sDn\ba}$ whose width and height are both $O(n\size A)$.
\end{lemma}

\begin{lemma}[restate = PhaseTwoMerge, name = ]
    \label{lemma:phase2merge}
Let $\atmenum{n}$ be an atomic enumeration, let $1\leq j\leq n$ and let $v$ be a variable. Then there exists a $\TBLScf$ derivation \[\odv{\expr{n}{v}}
{}{\range{\expri{n}{j}{v} }{l_j} \ba_j \range{\expri{n}{j}{v} }{r_j}}{\ba\sUp\ba_j}\] whose width and height are both $O(n)$.
\end{lemma}

\subsection{Phase III: Guarded substitutions}

\begin{figure*}[!t]
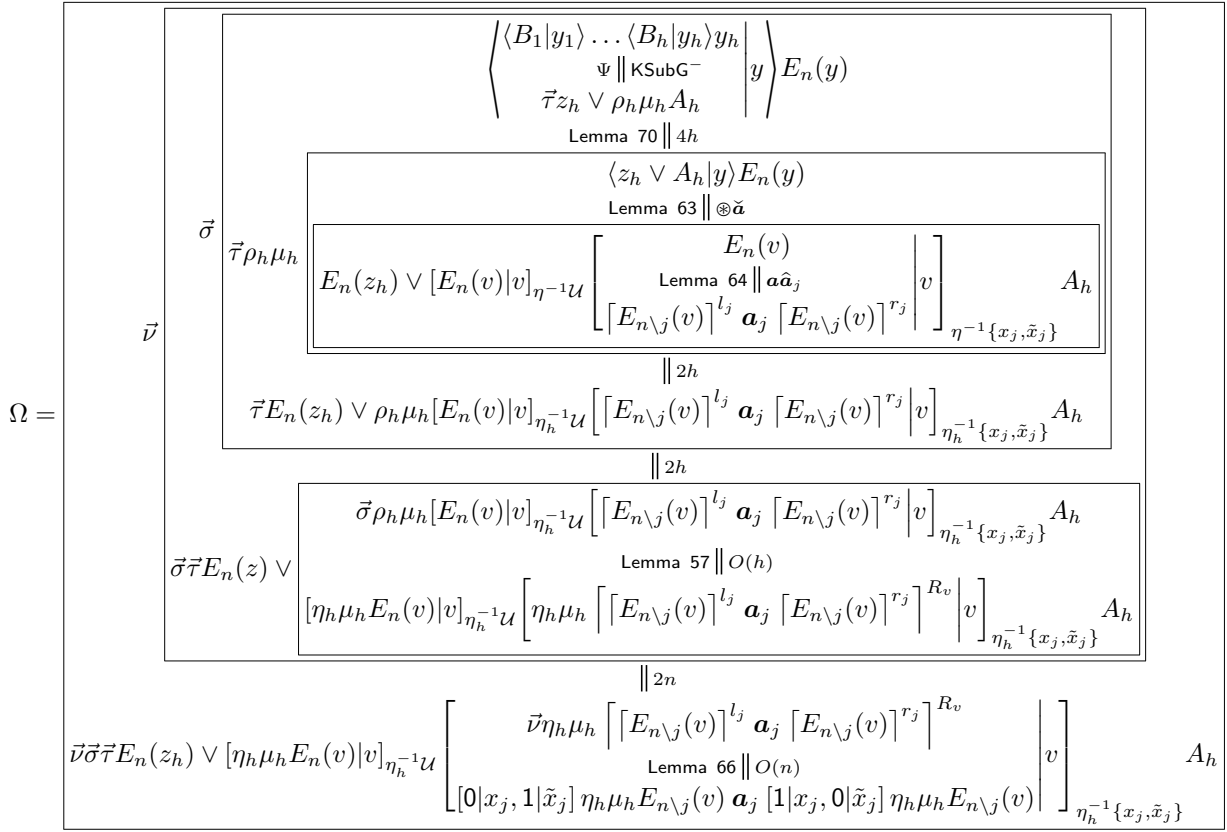

\[\dere=
\od{
\odv
{\vec\nu\; \odv
  {\vec\sigma\; \od{\odd{\odd{\odh
  {\esbs{\odframefalse\odv
    {\esbs{B_1}{y_1}\dots\esbs{B_h}{y_h}y_h  }
    {\derb}{[\vec\tau z_h . \rho_h\mu_hA_h] }{\TBLScf } }{y}\expr{n}{y} }}
  {\mathsf{Lemma~\ref{lemma:extracttau}}}{\vec\tau\rho_h\mu_h\; \odv
  {\esbs{z_h \vlor A_h}{y} \expr{n}{y} }
{\mathsf{Lemma~\ref{lemma:phase2eversion}}}{\odbox{\expr{n}{z_h} \vlor 
\acsnoun{\expr{n}{v}}{v}{\subpre{\eta}{\Unts} }
\acsnoun{\odframefalse\odv
  {\expr{n}{v}}
{\mathsf{Lemma~\ref{lemma:phase2merge}}}{\range{\expri{n}{j}{v} }{l_j} \ba_j 
   \range{\expri{n}{j}{v} }{r_j}   }{\ba\sUp\ba_j} }
{v}{\subpre{\eta}{\set{x_j,\varneg{x}{j}}} }
A_h }}{\alp\sDn\ba}}{4h}}
  {}{[\vec\tau\expr{n}{z_h} .\rho_h\mu_h
        \acsnoun{\expr{n}{v}}{v}{\subpre{\eta_h}{\Unts} }
        \acsnoun{\range{\expri{n}{j}{v} }{l_j} \ba_j 
                 \range{\expri{n}{j}{v} }{r_j}}{v}{\subpre{\eta_h}{\set{x_j,\varneg{x}{j}}} }A_h
        ]}{2h}} }
  {}{[\vec\sigma\vec\tau\expr{n}{z} .
  \odv{\vec\sigma\rho_h\mu_h 
  \acsnoun{\expr{n}{v}}{v}{\subpre{\eta_h}{\Unts} }
        \acsnoun{\range{\expri{n}{j}{v} }{l_j} \ba_j 
                 \range{\expri{n}{j}{v} }{r_j}}{v}{\subpre{\eta_h}{\set{x_j,\varneg{x}{j}}} }A_h}
    {\mathsf{Lemma~\ref{lemma:sigmarhoeta}}}{\acsnoun{\eta_h\mu_h \expr{n}{v}}{v}{\subpre{\eta_h}{\Unts} }
       \acsnoun{\eta_h\mu_h\range{\range{\expri{n}{j}{v} }{l_j} \ba_j 
                 \range{\expri{n}{j}{v} }{r_j} }{R_v} }{v}{\subpre{\eta_h}{\set{x_j,\varneg{x}{j}}}}
       A_h }{O(h) }
  ] }{2h}}
{}{[\vec\nu\vec\sigma\vec\tau \expr{n}{z_h} .
\acsnoun{\eta_h\mu_h\expr{n}{v}}{v}{\subpre{\eta_h}{\Unts} }
\acsnoun{\odframefalse\odv
         {\vec\nu\eta_h\mu_h\range{\range{\expri{n}{j}{v} }{l_j} \ba_j 
                 \range{\expri{n}{j}{v} }{r_j} }{R_v} }
         {\mathsf{Lemma~\ref{lemma:applynu}} }{\asbss{\zer}{x_j}{\one}{\varneg{x}{j} }\eta_h\mu_h\expri{n}{j}{v} \ba_j 
             \asbss{\one}{x_j}{\zer}{\varneg{x}{j} }\eta_h\mu_h\expri{n}{j}{v}}{O(n)} }
         {v}{\subpre{\eta_h}{\set{x_j,\varneg{x}{j}}} }
A_h
]}{2n}
}
\]
    \caption{A $\TBLScf$ proof which p-simulates the $\sFzo$ proof $\dera$. Here $\vec\nu$, $\vec\sigma$, and $\vec\tau$ abbreviate $\nu_n\dots\nu_1$, $\sigma_h\dots\sigma_1$, and $\tau_1\dots\tau_h$ respectively.}
    \label{figure:fregesim}
\end{figure*}

Lemmas \ref{lemma:phase2eversion} and \ref{lemma:phase2merge} above show that the formula $\esbs{z_h\vlor \eta_h\mu_hA_h}{y}\expr{n}{y}$ can be transformed into one with free variables $x_j\rng{l_j}$, $x_j\rng{r_j}$, $\varneg x j \rng{l_j}$, and $\varneg x j \rng{r_j}$ for each $j\in\set{1,\dots,n}$; that is, a formula to which the guarded substitutions $\nu_1,\dots,\nu_n$ can be applied. In Lemma~\ref{lemma:applynu} we show that indeed by applying these guarded substitutions we can obtain a formula whose interpretation is the conclusion of the given $\sFzo$ proof. 

\begin{definition}\label{def:substpreimage}
Let $V$ be a set of variables and let $\eta=\acsnoun{C_j}{v_j}{V}$ be an elementary substitution. The equivalence relation on $V$ given by $\eta v = \eta v'$ induces a partition on $V$, and for any $W\subseteq \Unts\cup\Vbls^\Rngs$, we write $\subpre \eta W$ for the set $\set{v\in V \mid \eta v \in W}$.
\end{definition}

\begin{lemma}[restate = PhaseThree, name = ]
    \label{lemma:applynu}
For each $j\in\set{1,\dots,n}$, $\mu_h$ as given in Construction \ref{cons:phase-1-a}, and $v\in\subpre{\eta_h}{\set{x_j,\varneg x j }} $, there exists a derivation \[\odv
         {\nu_n\dots\nu_1\eta_h\mu_h(({\range{\expri{n}{j}{v} }{l_j} \ba_j 
                 \range{\expri{n}{j}{v} }{r_j} }))}
         {\derc(v) }{\asbss{\zer}{x_j}{\one}{\varneg{x}{j} }\eta_h\mu_h\expri{n}{j}{v} \ba_j 
             \asbss{\one}{x_j}{\zer}{\varneg{x}{j} }\eta_h\mu_h\expri{n}{j}{v}}{}\] 
whose height and width are each $O(n)$ and whose only vertical composition is by expansion.

Furthermore, if $\eta_h v=x_j$ then $\intof{\cn\derc(v)}=a_j$ and if $\eta_h=\varneg x j$ then $\intof{\cn\derc(v)}=\atmneg{a}{j}$.
\end{lemma}

\subsection{Putting everything together}

We can now combine phases I, II, and III to prove our main result. The full construction is shown in 
Figure \ref{figure:fregesim}. When we write an actual substitution $\acsnoun{\derd}{v}{\subpre{\eta_h}{\set{x_j,\varneg x j}}}$, we understand this to be indexed over $j\in \set{1,\dots,n}$, so that this abbreviates the substitution $\acsnoun{\derd}{v}{\subpre{\eta_h}{\set{x_1,\varneg x 1}}}\dots \acsnoun{\derd}{v}{\subpre{\eta_h}{\set{x_n,\varneg x n}}}$.

\begin{theorem}[restate = MainFregezo, name = ]
  \label{thm:main-sfzo}
$\TBLScf$ p-simulates $\sFzo$.
\end{theorem}

Now Theorem~\ref{thm:main} follows immediately.


\section{What are Guarded Substitutions?}
\label{sec:guarded}

To conclude this paper, we turn to a more speculative discussion of what guarded substitutions really are, and what we might want to do with them in the future. Here, we used the concept of guarded substitution only in very restricted way: we only discussed propositional logic, and we only allowed the units $\zer$ and $\one$ to be substituted. However, the notion of guarded substitution has a much larger potential. To illustrate this, we give an informal example of how guarded substitutions could be used for compression purposes in a familiar proof.

We take the well-known \cite{Maor:19:pythagorean} proof of Pythagoras' theorem using similar triangles, which, semi-formally, we can represent concisely in a first-order proof with guarded substitutions.

The proof begins by establishing that $ACP$ and $BCP$ are congruent to the triangle $ABC$, and hence shows that $\size{AB}^2=\size{AC}^2+\size{BC^2}$.
These two proofs of congruency 
are similar but not identical, differing only in that $A$ and $B$ are transposed. 
\begin{wrapfigure}[7]{l}{.18\textwidth}
  \tikz{\coordinate[label=left : {$A$}] (A) at (0,0);
  \coordinate[label=right: {$B$}] (B) at (2.5,0);
  \coordinate[label=above: {$C$}] (C) at (1.6,1.2);
  \coordinate[label=below: {$P$}] (P) at (1.6,0);
  \draw (A) -- (B)
        (A) -- (C)
        (B) -- (C)
        (C) -- (P);}
\end{wrapfigure}
This can give rise to a superposition: a derivation $\derc$ with free variables $x$ and $y$ whose conclusion is that triangles $xCP$ and $xyC$ are congruent. 
Applying a substitution $\asbss AxBy$ results in a proof that triangles $ACP$ and $ABC$ are congruent, and so the guarded substitutions in the proof $\gsbss AxBy1\gsbss BxAy2 \esbs\derc z (z\rng 1\vlan z\rng 2)$ can be applied in order to recover the conclusion that both the triangles $ACP$ and $BCP$ are congruent to the triangle $ABC$, and the proof of Pythagoras' theorem can proceed as normal from there. 

This proof offers us a glimpse of how guarded substitutions might lead the way towards a proof system that offers both proof compression while still retaining a notion of analyticity; and not only free composition of derivations by connective, but free substitution of derivations too \cite{Gugl:04:Formalis:ea}. Subatomic proof systems\textemdash where strict linearity allows us to restrict all duplication to expansion of substitutions rather than inference rules\textemdash offer a sterile laboratory for developing these ideas, but we expect that these concepts can, with care, be translated into familiar proof systems and eventually programming languages, while maintaining in some form the proof theoretic properties established in this paper.



\newpage
\bibliography{IEEEabrv,swc-bib}


\newpage
\appendix

\subsection{Proofs and definitions for Section~\ref{sec:prelims}}

\begin{definition}\label{def:premise}
We inductively define the four maps \defn{premise} and \defn{conclusion}, ${\pr},{\cn}\colon\SADerv\to\Fmla$, and \defn{width} and \defn{height}, ${\w},{\h}\colon\SADerv\to\Nat$ as follows:
\begin{itemize}
\item If $\dera\in\Fmla$, then $\pr\dera\ideq\cn\dera\ideq\dera$ and $\w\dera=\size\dera$ and $\h\dera=0$
\item If $\dera\ideq \derb \alp \derc$ for some $\alp\in\Atms\cup\Conn$, then 
\[
\begin{aligned}
    \pr\dera&\ideq \pr\derb \alp \pr\derc & \w\dera&=\w\derb+\w\derc\\
    \cn\dera&\ideq \cn\derb \alp \cn\derc & \h\dera&=\max(\h\derb,\h\derc) \\
\end{aligned}
\]
\item If $\dera\ideq\esbs    \derb x    \derc$, then
\[
\begin{aligned}
    \pr\dera&\ideq\esbs{\pr\derb}x{\pr\derc} & \w\dera&=\w\derb+\w\derc\\
    \cn\dera&\ideq\esbs{\cn\derb}x{\cn\derc} & \h\dera&=\max(\h\derb,\h\derc)\\
\end{aligned}
\]
\item If $\dera\ideq\gsbs    \derb x  p  \derc$, then
\[
\begin{aligned}
    \pr\dera&\ideq\gsbs{\pr\derb}x p{\pr\derc} & \w\dera&=\w\derb+\w\derc\\
    \cn\dera&\ideq\gsbs{\cn\derb}x p{\cn\derc} & \h\dera&=\max(\h\derb,\h\derc)\\
\end{aligned}
\]
\item if $\dera\ideq\vls\odn{\;\derb\;}{}\derc{}\,$ or $\dera\ideq\vls\odq{\;\derb\;}{}\derc{}\,$, then 
\[
\begin{aligned}
    \pr\dera&\ideq\pr\derb & \w\dera&=\max(\w\derb,\w\derc)\\
    \cn\dera&\ideq\cn\derc & \h\dera&=\h\derb+\h\derc+1 \\
\end{aligned}
\]
\end{itemize}
\end{definition}

We now show some basic properties of composition by expansion, and we annotate the derivations constructed with their height. These are needed for later proofs.

The following lemma is used in Lemmas \ref{lemma:mergemedial} and \ref{lemma:mergeswitch}.
\begin{lemma}\label{lemma:mergeconstructor}
Let $A,B,C,D$ be subatomic formulae with no guarded substitutions, and let $y$ and $z$ be variables which occur freely in none of them. Then there exist invertible derivations \[\odv{\esbs{A}{x}B \alp \esbs{C}{x}D}
{}{\esbs{A}{y} \esbs{C}{z} ((\asbs{y}{x}B \alp \asbs{z}{x}D ))}{7} \qquad 
\odv{\esbs{A}{y} \esbs{C}{z} \asbs{y \alp z}{x}B }{}{\esbs{A \alp C}{x}B}{4}
\] 
\end{lemma}
\begin{proof}
We note that $D=\asbs xz\asbs zx D$ and $B=\asbs xy\asbs yx B$, and we build:
\[\odq{
\od{\ods{\ods{\odh
    {\esbs A x B}}
    {}{\esbs Ax \esbs xy \asbs yx B}{}}
    {}{\esbs Ay\asbs yx B}{}} \alp 
\od{\ods{\ods{\ods{\odh
    {\esbs CxD}}
    {}{\esbs Cx\esbs xz\asbs zxD}{}}
    {}{\esbs Cz\asbs zx D}{}}
    {}{\esbs Ay\esbs Cz\asbs zx D}{}} }
{}{\esbs Ay \odq
{\odq{\asbs yxB}{}{\esbs Cz\asbs yxB}{}\alp \esbs Cz\asbs zxD}
{}{\esbs Cz((\asbs yxB \alp \asbs zxD ))}{}}{}\]
for the derivation on the left and 
\[
\odq{\esbs Ay\odq{\esbs Cz\odq{\asbs{y\alp z}xB}{}{\esbs{y\alp z}xB}{} }{}{\esbs{y\alp C}xB}{}}
{}{\esbs{A\alp C}xB}{}
\] for the derivation on the right.
\end{proof}

The following lemma is used $h$ many times in the derivation $\dere$ in Figure \ref{figure:fregesim} after the first phase, in order to move the explicit substitutions $\tau_1,\dots,\tau_h$ out of the way.
\begin{lemma}\label{lemma:extracttau}
Let $C,D,A,E$ be subatomic formulae with no guarded substitutions, and let $z$ be a variable which occurs freely in neither $A$ nor $E$. Then there exists an invertible derivation
\[\odv{\esbs{\esbs CzD \alp A}yE}{}{\esbs Cz \esbs{D\alp A}yE }{4}\]
\end{lemma}
\begin{proof}
We note that $A=\asbs CzA$ and $E=\asbs CzE$, and we build:
\[
\odq{\esbs{\odframefalse\odq{\esbs CzD \alp \odq{A}{}{\esbs CzA}{} }{}{\esbs Cz(D\alp A)}{} }{y}\odq{E}{}{\esbs CzE}{} }
{}{\esbs Cz\esbs{D\alp A}yE }{}
\]
\end{proof}

\MergeMedial*

\begin{figure*}[!t]
  \[
\od{\odd{\odd{\odh
{\esbs{\acs{B_i}{x_i}{D} D}{y}\acsnoun{B_i}{x_i}{\fv E \setminus \set{y}}E \mathbin{\sDn\alp} 
\esbs{\acs{C_i}{x_i}{D} D}{y}\acsnoun{C_i}{x_i}{\fv E \setminus \set{y}}E}}
{\mathsf{Lemma~\ref{lemma:mergeconstructor}}}{\esbs{\acs{B_i}{x_i}{D} D}{y_1}\esbs{\acs{C_i}{x_i}{D} D}{y_2}\odv
{\asbs{y_1}{y}\acsnoun{B_i}{x_i}{\fv E \setminus y}E \mathbin{\sDn\alp} \asbs{y_2}{y}\acsnoun{C_i}{x_i}{\fv E \setminus y}E}
{\mathsf{IH}}{\asbs{y_1\mathbin{\sDn\alp} y_2}{y}\acsnoun{B_i\mathbin{\sDn\alp} C_i}{x_i}{\fv E \setminus y}E}{\sDn\alp\bet,\bet\sDn\alp}  }{7}}
{\mathsf{Lemma~\ref{lemma:mergeconstructor}}}{\esbs{\odframefalse\odv
{\acs{B_i}{x_i}{D} D \mathbin{\sDn\alp} \acs{C_i}{x_i}{D} D}
{\mathsf{IH}}{\acs{B_i\mathbin{\sDn\alp} C_i}{x_i}{D} D }{\sDn\alp\bet,\bet\sDn\alp}}{y}\acsnoun{B_i\mathbin{\sDn\alp} C_i}{x_i}{\fv E \setminus y}E  }{4}}
\]\caption{Explicit substitution case for the proof of Lemma \ref{lemma:mergemedial} }
    \label{figure:mergesubmedial}
\end{figure*}

\begin{proof}
We prove the derivation on the left by structural induction on $A$.
  
  If $A$ is a variable then the derivation is a formula.
  
  If $A=D\bet E$ for some $\bet\in\Conn\cup\Atms$ then we build \[
\odn{((\acs{B_i}{x_i}{D}D \bet \acs{B_i}{x_i}{E}E)) \mathbin{\sDn\alp} 
((\acs{C_i}{x_i}{D}D \bet \acs{C_i}{x_i}{E}E)) }{\sDn\alp\bet}
{\odv
{\acs{B_i}{x_i}{D}D \mathbin{\sDn\alp} \acs{C_i}{x_i}{D}D}
{\mathsf{IH}}{\acs{B_i\mathbin{\sDn\alp} C_i}{x_i}{D}D}{\sDn\alp\bet,\bet\sDn\alp} \bet \odv
{\acs{B_i}{x_i}{E}E \mathbin{\sDn\alp} \acs{C_i}{x_i}{E}E}
{\mathsf{IH}}{\acs{B_i\mathbin{\sDn\alp} C_i}{x_i}{E}E}{\sDn\alp\bet,\bet\sDn\alp}
}{}
\]

If $A=\esbs{D}{y}E$ then for variables $y_1$, $y_2$ not appearing in $E$ we build the derivation in Figure \ref{figure:mergesubmedial}, and the case  $A=\gsbs{D}{y}pE$ is similar and simpler.

The derivation on the right can be constructed analogously.
\end{proof}

\MergeSwitch*
\begin{figure*}[!t]
  \[
\od{\odd{\odd{\odh
{\esbs{\odframefalse\odv
{\acs{B_i\alp C_i}{x_i}{D} \sUp D }
{\mathsf{IH}}
{\acs{B_i}{x_i}{D} D \alp \acs{C_i}{x_i}{D} -D}
{\sDn\bet\alp,\alp\sDn\bet}}{y}\acsnoun{B_i\alp C_i}{x_i}{\fv E \setminus y}\sUp E}}
{\mathsf{Lemma~\ref{lemma:mergeconstructor}}}{\esbs{\acs{B_i}{x_i}{D} D}{y_1}\esbs{\acs{C_i}{x_i}{D} -D}{y_2}\odv
{\asbs{y_1\alp y_2}{y}\acsnoun{B_i\alp C_i}{x_i}{\fv E \setminus y}\sUp E}
{\mathsf{IH}}{\asbs{y_1}{y}\acsnoun{B_i}{x_i}{\fv E \setminus y}E \alp \asbs{y_2}{y}\acsnoun{C_i}{x_i}{\fv E \setminus y}-E}{\sDn\bet\alp,\alp\sDn\bet}  }{4}}
{\mathsf{Lemma~\ref{lemma:mergeconstructor}}}{\esbs{\acs{B_i}{x_i}{D} D}{y}\acsnoun{B_i}{x_i}{\fv E \setminus \set{y}}E \alp
\esbs{\acs{C_i}{x_i}{D} -D}{y}\acsnoun{C_i}{x_i}{\fv E \setminus \set{y}}-E  }{7}}
\]\caption{Explicit substitution case for the proof of Lemma \ref{lemma:mergeswitch} }
    \label{figure:mergesubswitch}
\end{figure*}

\begin{proof}
We prove the derivation on the left by structural induction on $A$.

If $A$ is a variable then the derivation is a formula. 

If $A=D\bet E$ for some $\bet\in\Conn\cup\Atms$ then we can assume without loss of generality that $\sUp\bet=\bet$ and $\sDn\bet=\cneg\bet$ and we build:
\[
\odn{\odv{\acs{B_i \alp C_i }{x_i}{D} \sUp D}
{\mathsf{IH}}{\acs{B_i}{x_i}{D} D \alp
   \acs{C_i }{x_i}{D} -D}{\sDn\bet\alp,\alp\sDn\bet} \mathbin{\sUp\bet}
   \odv{\acs{B_i \alp C_i }{x_i}{E} \sUp E}
{\mathsf{IH}}{\acs{B_i}{x_i}{E} E \alp
   \acs{C_i }{x_i}{E} -E}{\sDn\bet\alp,\alp\sDn\bet}
   }
{\sDn{\bet}\alp}
{((\acs{B_i}{x_i}{D} D \bet \acs{B_i}{x_i}{E} E))\alp ((\acs{C_i }{x_i}{D} -D \mathbin{\cneg\bet} \acs{C_i }{x_i}{E} -E))}{}
\]

If $A=\esbs{D}{y}E$ then for variables $y_1$, $y_2$ not appearing in $E$ we build the derivation in Figure \ref{figure:mergesubswitch}, and the case  $A=\gsbs{D}{y}pE$ is similar and simpler.

The derivation on the right can be constructed analogously.
\end{proof}


\subsection{Proofs for Section~\ref{sec:sub-to-standard}}

\begin{definition}\label{def:SKS}
The proof system $\SKS$~\cite{Brun:04:Deep-Inf:rq,GuglGundPari::A-Proof-:fk} consists of the
following inference rules:
\begin{itemize}
\item The \defn{structural rules}:
\[\text{\emph{atomic}}\left\{\begin{array}{ccc}
\odn{\one}{\aiD}{\vls[a.-a]}{}  & \odn{\vls[a.a]}{\acD}{a}{}  & \odn{\zer}{\awD}{a}{} \\
\emph{\: identity} & \emph{\: contraction} & \emph{\, weakening} \\ 
\\[-1ex]
\odn{\vls(a.-a)}{\aiU}{\zer}{}   & \odn{a}{\acU}{\vls(a.a)}{}   & \odn{a}{\awU}{\one}{} \\
\emph{\: cut} & \emph{\: cocontraction}  & \emph{\, coweakening}
\end{array}\right\}\]
\item The \defn{logical rules}:
\[\begin{array}{ccc}
\odn{\vls(A.[B.C])}{\sw}{\vls[(A.B).C]}{}  & & \odn{\vls[(A.B).(C.D)]}{\me}{\vls([A.C].[B.D])}{} \\
\emph{\: switch} & & \emph {\: medial}
\end{array}\]
\item An equivalence on formulae, defined to be the minimal equivalence relation generated by the following:
\[\begin{aligned}
\vls(A.\one) &\fequiv A  &\vls(A.B) &\fequiv \vls(B.A)\\
\vls[A.\zer] &\fequiv A  & \vls[A.B] &\fequiv \vls[B.A] \\\\[-2ex]
\vls[\one.\one] &\fequiv \one \qquad&\vls(A.(B.C)) &\fequiv \vls((A.B).C) \\
\vls(\zer.\zer) &\fequiv \zer & \vls[A.[B.C]] &\fequiv\vls[[A.B].C] \\
\end{aligned}\]

We write $\odn{A}{\fequiv}{B}{}$ whenever $B$ can be reached from $A$ via this equivalence relation. 

\item We also include the mix rule:
\[\odn{\zer}{\mix}{\one}{}\]
to enable a straightforward definition of the interpretation function, but note that it is easily derivable from the other rules.
\end{itemize}
A \defn{proof} is a derivation in $\SKS$ with premise $\one$. And a standard formula $A$ is \defn{provable} if there is a proof with conclusion $A$. 
\end{definition}

\begin{theorem}\label{thm:SKS}
  The system $\SKS$ is sound and complete for classical propositional logic, i.e., a formula is provable in $\SKS$ iff it is a tautology. \rm \cite{brunnler:tiu:01,Brun:04:Deep-Inf:rq,GuglGundPari::A-Proof-:fk}
\end{theorem}

\begin{theorem}\label{thm:SKS-cut}
  The cut is admissible from  $\SKS$. More precisely, if a formula $A$ is provable in $\SKS$, then it is also provable without using the rules $\aiU$, $\acU$, $\awU$. \rm \cite{brunnler:tiu:01,Brun:04:Deep-Inf:rq}
  \end{theorem}

\PropDeri*

\begin{proof}
  It suffices to observe that if $\odn A{}B{}$ is an instance of an inference rule in $\TBLS$, and both $\intof A$ and $\intof B$ are defined, then $\odn {\intof A}{}{\intof B}{}$ is an instance of an inference rule in $\SKS$.
\end{proof}

\IntPoly*

\begin{proof}
By induction on the structure of $A$.

If $A\in\Vbls^\Rngs$ then $\intof{\esbs{B_1}{v_1}A}=\intof{B_1}=U$.

If $A=C\alp D$ then it follows from the definition of the interpretation map that $\intof{\exs{B_i}{v_i}AA} = \intof{\exs{B_i}{v_i}CC \alp \exs{B_i}{v_i}DD }$. By induction, we can assume that $\intof{\exs{B_i}{v_i}CC}=U $ and $\intof{\exs{B_i}{v_i}DD}=U$, so it follows that $\intof{\exs{B_i}{v_i}AA}=U$.

If $A=\esbs CxD$ then it follows from the definition of the interpretation map that $\intof{\exs{B_i}{v_i}AA}=\intof{\esbs{\exs{B_i}{v_i}CC }{x}\exsnoun{B_i}{v_i}{\fv D\setminus \set x}D  } $. By induction, we can assume that $\intof{\exs{B_i}{v_i}CC}=U$, and so it follows by induction that $\intof{\esbs{\exs{B_i}{v_i}CC }{x}\exsnoun{B_i}{v_i}{\fv D\setminus \set x}D  } =U$.

\end{proof}

\SoundCompl*

\begin{proof}
  1) $\Leftrightarrow$ 2) follows from Theorem~\ref{thm:SKS}. For 2) $\Rightarrow$ 3) observe that every instance of a rule of $\SKS$ can be seen as the interpretation of an instance of a subatomic rule. The cases of $\acD$ and $\acU$ and $\aiU$ were given in the introduction. The other cases are similar. The result has also been shown in~\cite{AlerGugl:18:Subatomi:ty}.
  3) $\Rightarrow$ 2) also has been shown in~\cite{AlerGugl:18:Subatomi:ty}. 3)  $\Rightarrow$ 4) is shown in the same way as 2) $\Rightarrow$ 3) because the interpretation of every instance of an algebraic rule is (if defined) also the interpretation of a subatomic rule instance. Finally,   4) $\Rightarrow$ 3) is trivial.
\end{proof}


\subsection{Proofs for Section \ref{sec:psim}}

\StToSub*

\begin{proof}
This follows immediate from the definition.
\end{proof}

Before we continue, we first show the remaining cases for Construction~\ref{cons:phase-1-a}:

\medskip

If $P_i$ is the axiom $\frF_2$, Then $A_i$ is of the form $\vls [(X.(Y.Z')).[(X'.Y').[X''.Z]]]$ where $\eta_i X =\eta_i X'$ and both are pseudo-dual to $\eta_i X''$, $\eta_i Y$ and $\eta_i Y'$ are pseudo-dual, and $\eta_i Z$ and $\eta_i Z'$ are pseudo-dual. For a variable $v\in \fv X$, we denote by $v'$ and $v''$ the variables in $\fv X'$ and $\fv X''$ respectively appearing in the corresponding position, and similarly for $Y$ and $Z$. With use of Lemmas \ref{lemma:mergemedial} and \ref{lemma:mergeswitch} we construct the derivation $\derb_i$ and formula $B_i$ in Figure \ref{figure:frege2}.
\begin{figure*}
\[\begin{array}{c}
B_i= \zer\vlor\eta_i(\acs{v\vlor v''}{v}{X}\sUp X\vlan (\acs{v'\vlor v''}{v}{X}\sUp X\vlan (\acs{v\vlor v'}{v}{Y}\sUp Y \vlan \acs{v'\vlor v}{v}{Z}\sUp Z) ) ) \\ \\
\Psi_i = \od{\odq{\zer}{}{\esbs{\zer}{z_i}z_i}{} \vlor \eta_i\,
\odv{\odbox{(\odv
{\acs{v\vlor v''}{v}{X} \sUp X}
{}{[X.X'']}{\sDn\vlan\vlor,\vlor\sDn\vlan}.\odn{(\odv
{\acs{v'\vlor v''}{v}{X}\sUp X}
{}{[X'.X'']}{\sDn\vlan\vlor,\vlor\sDn\vlan}.\odn{(\odv
{\acs{v\vlor v'}{v}{Y}\sUp Y}
{}{[Y.Y']}{\sDn\vlan\vlor,\vlor\sDn\vlan}.\odv
{\acs{v'\vlor v}{v}{Z}\sUp Z}
{}{[Z'.Z]}{\sDn\vlan\vlor,\vlor\sDn\vlan})}
{\vlor\sDn\vlan}{[(Y.Z').[Y'.Z]]}{})}
{\vlor\sDn\vlan}{[(X'.(Y.Z')).[X''.[Y'.Z]]]}{})}}
{}{\odbox{[(X.(Y.Z')).\odn
{([X'.X''].\odv
{[X''.[Y'.Z]]}
{}{[Y'.[X''.Z]]}{\assD,\comD})}
{\vlor\sDn\vlan}{[(X'.Y').\odn
{[X''.[X''.Z]]}
{\assD}{[\odv
{[X''.X'']}
{}{\acs{[v.v]}{v}{X''}X''}{\vlan\sDn\vlor,\vlor\sDn\vlor}.Z]}{}]}{}]}}{\sw,\assD,\comD}}
\end{array}
\]
\caption{$\frF_2$}
\label{figure:frege2}
\end{figure*}
We obtain the substitutions $\tau_i=\esbs{\zer}{z_i}$, $\rho_i=\eta_i$, $\mu_i=\acs{v\vlor v}{v}{X''}$, and $\sigma_i=\epsilon$.

\medskip

In the case of $\frF_3$, we have that $A_i$ is of the form $\vls [(X'.Y).[Y'.X]]$, where $\eta_i X'$ and $\eta_i X$ are pseudo-dual formulae, and so are $\eta_i Y$ and $\eta_i Y'$; for $v\in \fv X$, we denote by $v'$ the variable in $\fv X'$ appearing in the corresponding position, and similarly for $Y$. With use of Lemma \ref{lemma:mergeswitch} we construct:
\[\Psi_i = \od{\odq{\zer}{}{\esbs{\zer}{z_i}z_i}{} \vlor \eta_i\, 
\odn{
\odv{\acs{v'\vlor v}{v}{X} \sUp X }{}{X'\vlor X}{\sDn\vlan\vlor,\vlor\sDn\vlan}\vlan
\odv{\acs{v\vlor v'}{v}{Y} \sUp Y }{}{Y\vlor Y'}{\sDn\vlan\vlor,\vlor\sDn\vlan}
}
{\vlor\sDn\vlan}{[(X'.Y).\odn{[X.Y']}{\comD}{[Y'.X]}{}] }{}
}\]
and we obtain the substitutions $\tau_i=\esbs{\zer}{z_i}$, $\rho_i=\eta_i$, $\mu_i=\epsilon$ and $\sigma_i=\epsilon$, and the formula $B_i=\zer\vlor \eta_i(\acs{v'\vlor v}{v}{X} \sUp X \vlan \acs{v\vlor v'}{v}{Y} \sUp Y) $.

\medskip

In the case of $\frF_4$, then $\eta_i A_i=\one$. We construct 
\[\Psi_i = \od{\odq{\zer}{}{\esbs{\zer}{z_i}z_i}{} \vlor \eta_i A_i
}\] and we obtain the substitutions $\tau_i=\esbs{\zer}{z_i}$, $\rho_i=\eta_i$, $\mu_i=\epsilon$ and $\sigma_i=\epsilon$, and the formula $B_i=\zer\vlor\one$.

The following lemma shows how guarded substitutions can be applied to non-flat formulae; we will use it in Lemmas \ref{lemma:sigmarhoeta} and \ref{lemma:applynu} to show how the guarded substitutions $\vec\sigma$ and $\vec\nu$ are applied in Figure \ref{figure:fregesim} in the proof of Theorem \ref{thm:main}.

\PhaseOneTautology*
\begin{proof}
By induction on the structure of $\dera$.

If $P_i$ is obtained from an axiom $\frF_1$ then $B_i$ is a disjunction $\zer \vlor \eta_i(\acs{v\vlan \zer}{v}{Y}Y \vlor \acs{v'\vlor v}{v}{X}\sUp X )$, where for each $v\in \fv{X}\cup\fv Y$, $\eta_i v$ and $\eta_i v'$ are pseudo-dual variables with respect to $\varenum{\dera}$. If $P_i$ is obtained from an axiom $\frF_2$ or $\frF_3$ then $B_i$ is composed of disjunctions of variables which are pseudo-dual with respect to $\varenum{\dera}$ (in disjunction with $\zer$). If $P_i$ is obtained from an axiom $\frF_4$ then $B_i=\zer\vlor\one$. In each case (except for $\frF_4$, which is immediate), the interpretation of the formula is determined by disjunctions of pseudo-dual variables, and any substitution which dualises pseudo-dual variables will verify it. Therefore, $B_i$ is tautological with respect to $\varenum{\dera}$.

If $P_i$ is obtained from modus ponens on $P_k$ and $P_l= \vls -P_k \vlor P_i$ then $B_i=y_k\vlan y_l$. We assume that $\asbs{B_1}{y_1}\dots \asbs{B_k}{y_k}y_k$ and $\asbs{B_1}{y_1}\dots \asbs{B_l}{y_l}y_l$ are tautological with respect to $\varenum{\dera}$, so $\asbs{B_1}{y_1}\dots \asbs{y_k\vlan y_l}{y_i}y_i$ is tautological with respect to $\varenum{\dera}$.

If $P_i$ is obtained from substitution on $P_k$ then $B_i=y_k\rng i$. We assume that $\asbs{B_1}{y_1}\dots \asbs{B_k}{y_k}y_k$ is tautological with respect to $\varenum{\dera}$. It follows immediately that $\asbs{B_1}{y_1}\dots \asbs{y_k\rng i}{y_i}y_i$ is tautological with respect to $\varenum{\dera}$.
\end{proof}

\PhaseOneContradiction*
\begin{proof}
By induction on the structure of $\dera$.

If $P_i$ is an axiom then $\tau_i=\esbs{\zer}{z_i}$, so $\tau_1\dots \tau_i z_i$ is contradictory with respect to $\varenum{\dera}$.

If $P_i$ is obtained from modus ponens on $P_k$ and $P_l= \vls -P_k \vlor P_i$ then $\tau_i$ is $\esbs{(z_k\vlor z_l) \vlor \acs{\rho_k\mu_k v \vlan \rho_l\mu_l v' }{v}{A_k}\sDn {A}_k}{z_i}$, where $\rho_k v$ and $\rho_l v'$ are pseudo-dual variables. We assume that $\tau_1\dots \tau_k z_k$ and $\tau_1\dots \tau_l z_l$ are contradictory with respect to $\varenum{\dera}$. By construction, $\rho_k\mu_k v \vlan \rho_l\mu_l v'$ is contradictory with respect to $\varenum{\dera}$ for each $v\in\fv{A_k}$. Therefore $\tau_1\dots \tau_i z_i$ is contradictory with respect to $\varenum{\dera}$.

If $P_i$ is obtained from substitution on $P_k$ then $\tau_i$ is $\esbs{z_k\rng i}{z_i}$. We assume that $\tau_1\dots \tau_k z_k$ is contradictory with respect to $\varenum{\dera}$. It follows immediately that $\tau_1\dots \tau_i z_i$ is contradictory with respect to $\varenum{\dera}$.
\end{proof}

\SigmaRhoEta*
\begin{proof}
By induction on the structure of $\dera$

If $P_i$ is an axiom then by construction, $\rho_i=\eta_i$, and the range of each variable in $\rho_i\mu_i A_i$ is empty, so $R=\emptyset$. We construct the derivation by vacuously applying each guarded substitution.

If $P_i$ is obtained from modus ponens on $P_k$ and $P_l= \vls -P_k \vlor P_i$ then for each $v$ occurring in $A_i$ there is $v'$ occurring in $A_l$ such that $\mu_l v' = \mu_i v$, $\rho_l = \rho_i$, and $\eta_l = \eta_i$. By induction, we can construct \[\sigma_h\dots\sigma_{l+1}\od{\odv{\sigma_l\dots\sigma_1\rho_i\mu_i v}{}{\eta_i\mu_i v^R}{O(l)}}\] and we have that $R\subseteq \set{1,\dots,l}$. Therefore, all of the substitutions $\sigma_{l+1},\dots\sigma_h$ can be vacuously applied to $\eta_i\mu_i v^R$ in $(i-l)$ many steps.

If $P_i$ is obtained from a substitution $\asbs{u_{j_1}}{x_{j_1}}\dots\asbs{u_{j_m}}{x_{j_m}}$ on $P_k$ then 
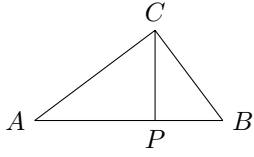
\[\sigma_i=\gsbss{u_{j_1}}{x_{j_1}}{\bar u_{j_1}}{\varneg{x}{j_1}}{i}\dots\gsbss{u_{j_m}}{x_{j_m}}{\bar u_{j_m}}{\varneg{x}{j_m}}{i}\] and for each $v$ occurring in $A_i$, there is a corresponding $w$ occurring in $A_k$ such that \[
\begin{aligned}
\mu_i v &= \mu_k w \\
\rho_i \mu_i v &= \rho_k\mu_k w\rng i \\
\eta_i &= \asbss{u_{j_1}}{x_{j_1}}{\bar u_{j_1}}{\varneg{x}{j_1}}\dots\asbss{u_{j_m}}{x_{j_m}}{\bar u_{j_m}}{\varneg{x}{j_m}}\eta_k
\end{aligned}
\]
By induction, there exists a derivation \[\odv{\sigma_k\dots\sigma_1\rho_k\mu_k w }{}{\eta_k\mu_k w^R}{O(k)} \] for some $R\subseteq\set{1,\dots,k}$. By construction, none of the guarded substitutions $\sigma_1,\dots,\sigma_{i-1}, \sigma_{i+1},\dots,\sigma_h$ have a guard $i$, so it follows that we can construct \[\odv{\sigma_h\dots\sigma_{i+1}\odq
{\sigma_i \odv
{\sigma_{\vlnos i-1}\dots\sigma_{k+1}\odv
{\sigma_k\dots\sigma_1 \rho_i\mu_i v}
{}{\eta_k\mu_i v^{R\cup\set i} }{O(k)} }
{}{\eta_k\mu_i v^{R\cup\set i} }{\vlnos (i-1)-k }}
{}{\eta_i\mu_i v^{R\cup\set i}}{} }
{}{\eta_i\mu_i v^{R\cup\set i} }{\vlnos h-i}\]
\end{proof}

\PsiIPoly*

\begin{proof}
We note that the height of $\dera$ is $h$. Its width, $w$, is given by $\max{\size{P_i}}$. The number of (positive and negative) atoms $2n$ is $O(hw)$.

The worst case for the width of a $\derb_i$ is obtained in the modus ponens and unit substitution cases when there have been many unit substitution steps, so that $\rho_i$ is substituting $O(h)$ many ranges onto each of the $O(w)$ many variables in $A_i$; the width of $\derb_i$ is then $O(hw)$.

The worst case for the height of a $\derb_i$ is obtained in the axiom and modus ponens cases: it is $O(w)$.
\end{proof}

\PsiPoly*
\begin{proof}
The size of $\derb$ can be calculated as follows:

\begin{align*}
    \w\derb &\leq h(\max \size{\tau_i} + 2\max\size{z_i\vlor \rho_i\eta_i A_i} ) \\
    & = hO(w + hw  ) \\
    & = O(h^2 w) \\
    \h\derb &= hO(\max\size {\derb_i} ) \\
    &= O(hw )
\end{align*}
\end{proof}

\Expprestaut*
\begin{proof}
If $A$ is a tautology with respect to $\varenum{n}$, then so is $\range A R$, for $R\in\set{l_1,r_1}\times\cdots\times\set{l_n,r_n}$. Also, for each such $R$, the actual actual substitution obtained by applying $\nu_n\dots\nu_1$ to $\range{A}{R}$ closes the formula and dualises pseudo-dual variables, so $\intof{\nu_n\dots\nu_1 \range A R}=\one$. Applying all of the explicit substitutions in $\esbs{A}{y} \expr{n}{y}$ results in a formula composed of $2^n$ copies of $A$, each with a different range from $\set{l_1,r_1}\times\cdots\times\set{l_n,r_n}$. It follows that $\intof{\nu_n\dots\nu_1 \esbs{A}{y} \expr{n}{y}}=\one$. The case for a contradictory formula is analogous.
\end{proof}

\PhaseTwoEversion*
\begin{proof}
By induction on $n$.

If $n=1$ then $\expr{1}{y}=\esbs{y\rng{l_1}\ba_1 y\rng{r_1}}{y}y$. By Lemma \ref{lemma:mergemedial} we can construct \[\od{\odd{\ods{\odh
{\esbs{A}{y} \esbs{y\rng{l_1}\ba_1 y\rng{r_1}}{y}y  }}
{}{\esbs{\odframefalse\odv
{\range{A}{l_1}\ba_1\range{A}{r_1} }
{}{\acs{{v\rng{l_1}\ba_1 v\rng{r_1}}}{v}{A}A}{\alp\sDn\ba_1 }}{y} y}{}}
{}{\acs{\expr{1}{v}}{v}{A}A }{2\size A}}\] and the height and width of this derivation are both $O(\size A)$.

If $n=k+1$ then $\expri{k+1}{1}{y}=\esbs{y\rng{l_2}\ba_2 y\rng{r_2}}{y}\dots\esbs{y\rng{l_{k+1}}\ba_{k+1} y\rng{r_{k+1}}}{y}y$. By induction and Lemma \ref{lemma:mergemedial} we can construct 
\[
\od{\odd{\odd{\ods{\odh
{\esbs{A}{y}\esbs{y\rng{l_1}\ba_1 y\rng{r_1}}{y}\expri{k+1}{1}{y} }}
{}{\esbs{\odframefalse\odv
{\range{A}{l_1}\ba_1\range{A}{r_1} }
{\mathsf{Lemma~\ref{lemma:mergemedial}}}{\acs{{v\rng{l_1}\ba_1 v\rng{r_1}}}{v}{A}A}{\alp\sDn\ba_1} }{y} \expri{k+1}{1}{y} }{}}
{}{\exs{v\rng{l_1}\ba_1 v\rng{r_1}}{v}{A}\odv
{\esbs{A}{y}\expri{k+1}{1}{y} }
{\mathsf{IH}}{\acs{\expri{k+1}{1}{v} }{v}{A}A }{\alp\sDn\ba }  }{2\size A}}
{}{\acs{\expr{k+1}{v} }{v}{A}A}{2\size A}}
\]
The height of the application of Lemma \ref{lemma:mergemedial} is $O(\size A)$; therefore the height of this derivation is $O(n\size A)$. Its width is given by the $O(\size A) n$ many explicit substitution terms which are factored out.
\end{proof}

\PhaseTwoMerge*
\begin{proof}
Let \[\expr{j-1}{v}= \esbs{v\rng{l_1}\ba_1 v\rng{r_1}}{v}\dots\esbs{v\rng{l_{j-1}}\ba_{j-1} v\rng{r_{j-1}}}{v}v\]
\[E'=
\esbs{v\rng{l_{j+1}}\ba_{j+1} v\rng{r_{j+1}}}{v}\dots\esbs{v\rng{l_n}\ba_n v\rng{r_n}}{v}v
\]
Then by making use of the construction 
\[
\odq{\esbs Ax((x\rng p \conna x\rng q))}
{}{\odq{\range A p}{}{\esbs{\range A p}xx }{} \conna 
   \odq{\range A q}{}{\esbs{\range A q}xx}{}}{}
\] to pass the ranges $l_j$ and $r_j$ to the free variables and applying Lemma \ref{lemma:mergemedial} we construct the following derivation:
\[
\od{\odd{\odd{\odh
{\expr{n}{v} }}
{}{\odq{\esbs{\range{\expr{\vlnos j-1}{v}}{l_j} \ba_j \range{\expr{\vlnos j-1}{v}}{r_j} }{v}E'}{}{\asbs{\range{\expr{\vlnos j-1}{v}}{l_j} \ba_j \range{\expr{\vlnos j-1}{v}}{r_j} }{v}E' }{}}{2j}}
{\mathsf{Lemma~\ref{lemma:mergemedial}}} 
  {\odq
  {\asbs{\range{\expr{\vlnos j-1}{v} }{l_j}}{v}E'}
  {}{\range{\expri{n}{j}{v}}{l_j} }{} \ba_j \odq
  {\asbs{\range{\expr{\vlnos j-1}{v} }{\vphantom{l_j}r_j}}{v}E'}
  {}{\range{\expri{n}{j}{v}}{\vphantom{l_j}r_j} }{}}{\ba\sUp\ba_j }}
\]

The height of this derivation is dominated by the application of Lemma \ref{lemma:mergemedial}, which has height $O(n)$.
\end{proof}

\begin{figure*}[!ht]
\[\tikz{\coordinate[label=left : {$A$}] (A) at (0,0);
  \coordinate[label=right: {$B$}] (B) at (2.5,0);
  \coordinate[label=above: {$C$}] (C) at (1.6,1.2);
  \coordinate[label=below: {$P$}] (P) at (1.6,0);

  \draw (A) -- (B)
        (A) -- (C)
        (B) -- (C)
  (C) -- (P);}
\hskip4em
\od{
\odi{\odi{\ods{\ods{\odh
{\gsbss{A}{x}{B}{y}{1}\gsbss{B}{x}{A}{y}{2}\esbs{\odframefalse\odn{\angle xCP = \frac{\pi}{2}\vlnos-\angle yxC = \angle Cyx }{}{\odframefalse\odn
    {\triangle Cyx \sim \triangle xCP}
  {}{\size{xC}^2=\size{xy}\size{xP}  }{} }{}}{z}(z\rng 1.z\rng 2) }}
{} {\gsbss{A}{x}{B}{y}{1}\gsbss{B}{x}{A}{y}{2}(\range{\size{xC}^2=\size{xy}\size{xP}}{1} .\range{\size{xC}^2=\size{xy}\size{xP}}{2} ) }{}}
{}{((\size{AC}^2=\size{AB}\size{AP}).(\size{BC}^2=\size{BA}\size{BP}))}{}}
{}{\size{AC}^2+\size{BC}^2=\size{AB}\size{AP}+\size{BA}\size{BP}}{}}
{}{\size{AC}^2+\size{BC}^2=\size{AB}^2}{}}
\] 
\caption{A fragment of a proof of Pythagoras' Theorem using guarded substitutions}
\label{fig:pythagoras}
\end{figure*}

\PhaseThree*
\begin{proof}
The only free variables in $\eta_h\mu_h\range{\expri{n}{j}{v} }{l_j} \ba_j 
                 \range{\expri{n}{j}{v} }{r_j} $
are either $x_j\rng{l_1,l_j}$, $x_j\rng{l_2,l_j}$, $x_j\rng{l_1,r_j}$, and $x_j\rng{r_1,r_j}$ or $\varneg x j \rng{l_1,l_j}$, $\varneg x j\rng{l_2,l_j}$, $\varneg x j\rng{l_1,r_j}$, and $\varneg x j\rng{r_1,r_j}$. The guarded substitutions $\eta_1,\dots,\eta_{j-1}$ do not apply to the variables $x_j$ or $\varneg x j$, and so they are applied vacuously. The substitution $\nu_j=\gsbss{\zer}{x_j}{\one}{\varneg x j}{j}\gsbss{\one}{x_j}{\zer}{\varneg x j}{j}$ can then be applied to $\eta_h\mu_h\range{\range{\expri{n}{j}{v} }{l_j} \ba_j \range{\expri{n}{j}{v} }{r_j} }{R_v}$, resulting in $\asbss{\zer}{x_j}{\one}{\varneg{x}{j} }\eta_h\mu_h\expri{n}{j}{v} \ba_j \asbss{\one}{x_j}{\zer}{\varneg{x}{j} }\eta_h\mu_h\expri{n}{j}{v}$. This formula has no free variables, so the guarded substitutions $\nu_{j+1},\dots,\nu_n$ are applied vacuously.

If $\eta_h v=x_j$ then $\intof{\asbss{\zer}{x_j}{\one}{\varneg{x}{j} }\eta_h\mu_h v }=\zer$, and it is immediate that $\intof{\asbs{\zer}{v}\expri n j v }=\zer$. Similarly, we can see that  $\intof{\asbss{\one}{x_j}{\zer}{\varneg{x}{j} }\eta_h\mu_h \expri n j v }=\one$; therefore $\intof{\cn\derc(v)}=\intof{\zer\ba_j\one}=a_j$. Analogously, if $\eta_h v=\varneg x j$ then $\intof{\cn\derc(v)}=\intof{\zer\ba_j\one}=\atmneg a j$.
\end{proof}

\MainFregezo*

\begin{proof}
Let $\dera=P_1,\dots,P_h$ ben an $\sFzo$ proof, with atomic enumeration $\atmenum{n}$, variable enumeration $\varenum{n}$, and factorised translation $(\eta_1,A_1),\dots,(\eta_h,A_h)$. 

Let $\derb$ and its constituent parts be as given in Constructions \ref{cons:phase-1-a} and \ref{cons:phase-1-b}, and let $\vec \tau$, $\vec \sigma$, and $\vec \nu$ abbreviate $\tau_1\dots\tau_h$, $\sigma_h\dots\sigma_1$, and $\nu_n\dots\nu_1$ respectively. We construct the proof in Figure \ref{figure:fregesim}, which we call $\dere$. 

First, we construct the ``superposition'' derivation $\derb$ given in Construction \ref{cons:phase-1-b}. We then extract the substitutions $\tau_1\dots\tau_h$, and use Lemmas \ref{lemma:phase2eversion} and \ref{lemma:phase2merge} to rearrange the formula. We distribute the substitutions $\sigma_h\dots\sigma_1$ and apply them. By Lemma \ref{lemma:sigmarhoeta}, these work as intended, applying the substitution inference steps in the $\sFzo$ proof $\dera$. In particular, for all $v\in \subpre {\eta_h}\Unts$, $\eta_h\mu_h \expr n v$ is a closed formula. Finally, we distribute the substitutions $\nu_n\dots\nu_1$ and apply them; by Lemma \ref{lemma:applynu} these work as intended, restoring the correct arguments to the atomic connectives. 

The height of this proof is dominated by the use of Lemma \ref{lemma:phase2eversion}, which is $O(nw)=O(hw^2)$.

The width of this proof is $O(hw(h+w))$. In the first phase, the width of the proof is dominated by the width of $\derb$, which is $O(h^2w)$. After the use of Lemma \ref{lemma:phase2eversion} and before the application of the substitutions $\sigma_h\dots\sigma_1$, the width of the proof is $O(hw^2)$.

By Lemma \ref{lemma:tautology}, $\asbs{B_1}{y_1}\dots\asbs{B_h}{y_h}$ is tautological with respect to $\varenum{n}$. By construction, applying the substitutions $\sigma_h\dots\sigma_1$ to $\asbs{B_1}{y_1}\dots\asbs{B_h}{y_h}$ does nothing but dualise some pairs of pseudo-dual variables, and so the resulting formula is also tautological with respect to $\varenum{n}$ . Then by Proposition \ref{prop:expprestaut}, $\intof{\nu_n\dots\nu_1\sigma_h\dots\sigma_1\esbs{B_1}{y_1}\dots\esbs{B_h}{y_h}\expr{n}{y} }=\one$.

By Lemma \ref{lemma:contradictory}, $\tau_1\dots\tau_h z_h$ is contradictory with respect to $\varenum{n}$. Then the flat formula obtained by applying all of the substitutions $\tau_1,\dots,\tau_h$ is also contradictory with respect to $\varenum{n}$. By construction, applying the substitutions $\sigma_h\dots\sigma_1$ to that formula does nothing but dualise some pairs of pseudo-dual variables, and so the resulting flat formula is also contradictory with respect to $\varenum{n}$. Then by Proposition \ref{prop:expprestaut}, $\intof{\nu_n\dots\nu_1\sigma_h\dots\sigma_1\tau_1\dots\tau_h\expr{n}{z_h} }=\zer$.

Furthermore, for all $v\in \subpre{\eta_h}{\Unts}$, \[\intof{\eta_h\mu_h\expr{n}{v} }=\intof{\eta_h v}\] for all $v\in \subpre{\eta_h}{\set{x_j}}$, 
\begin{align*}
    & \intof{\asbss{\zer}{x_j}{\one}{\varneg{x}{j} }\eta_h\mu_h\expri{n}{j}{v} \ba_j \asbss{\one}{x_j}{\zer}{\varneg{x}{j} }\eta_h\mu_h\expri{n}{j}{v}   } \\
= & \intof{\asbs{\zer\ba_j\one}{\eta_h} v} \\
= & a_j
\end{align*} and for all $v\in \subpre{\eta_h}{\set{\varneg x j}}$,
\begin{align*}
    & \intof{\asbss{\zer}{x_j}{\one}{\varneg{x}{j} }\eta_h\mu_h\expri{n}{j}{v} \ba_j \asbss{\one}{x_j}{\zer}{\varneg{x}{j} }\eta_h\mu_h\expri{n}{j}{v}   } \\
= & \intof{\asbs{\one\ba_j\zer}{\eta_h} v} \\
= & \atmneg{a}{j}
\end{align*}

Therefore, the interpretation of the conclusion is $P_h$ and by Proposition \ref{prop:intpoly}, this can be checked in polynomial time on the size of $\dera$..

\end{proof}

\MainThm*
\begin{proof}
  Immediately from Theorems~\ref{thm:sfzo} and~\ref{thm:main-sfzo}.
\end{proof}


\subsection{Examples for Section~\ref{sec:psim}}

\begin{example}\label{example:fregeaamp}
Let $\Phi$ be a $\sFzo$ proof given by
\[
\begin{array}{rlc}
P_1 &= \bar a \vlor ( b\vlor a) & (\frF_1) \\
P_2 &= (a\vlan (\bar b \vlan \bar a )) \vlor ((a \vlor b) \vlor (\bar a \vlor a) ) & (\frF_2) \\
P_3 &= (a \vlor  b) \vlor (\bar a \vlor a) & (\mp(P_1,P_2)) \\
\end{array}
\]

This has atomic enumeration $\atmenum{\Phi}=\{a,\bar a,b,\bar b\}$ (we use $a,b$ instead of $a_1,a_2$ to improve readability), and we take the variable enumeration $\varenum{\Phi}=\{w,\vneg w, x,\vneg x\}$ (likewise, we use $w,x$ instead of $x_1,x_2$). With respect to this variable enumeration, we obtain the open translations of the formulae of $\Phi$:

\[
\begin{array}{rl}
    Q_1 &= \vneg w{} \vlor (x\vlor w) \\
    Q_2 &= (w\vlan (\vneg x \vlan \vneg w )) \vlor ((w \vlor x) \vlor (\vneg w \vlor w) ) \\
    Q_3 &= (w \vlor x) \vlor (\vneg w \vlor w) \\
\end{array}
\]

We factorise these formulae to obtain the factorised open translations of the formulae of $\Phi$:

\[
\begin{array}{rl}
    A_1 &= v_1 \vlor (v_2\vlor v_3) \\
    \eta_1 &= \mathop{\left[\vneg w \middle| v_1, x \middle| v_2, w \middle| v_3 \right]} \\
    A_2 &= (v_4\vlan (v_5 \vlan v_6 )) \vlor ((v_7 \vlor v_8) \vlor (v_9 \vlor v_{10}) ) \\
    \eta_2 &= \mathop{\left[w\middle| v_4, \vneg x \middle| v_5, \vneg w \middle| v_6, w\middle| v_7,  x\middle| v_8, \vneg w\middle| v_9, w\middle| v_{10} \right]} \\
    A_3 &= (v_{11} \vlor v_{12}) \vlor (v_{13} \vlor v_{14}) \\
    \eta_3 &= \mathop{\left[w\middle| v_{11}, x\middle| v_{12}, \vneg w\middle| v_{13}, w\middle| v_{14} \right]} \\
\end{array}
\]

We obtain $\Psi_i, \tau_i, \mu_i, \rho_i, \sigma_i, B_i$ for $1\leq i\leq 3$ as follows:

\begin{itemize}
\item $\tau_1=\esbs \zer {z_1}$,
\item $\rho_1=\eta_1$,
\item $\mu_1=\asbs{x\vlor \zer}{v_2}$,
\item $\sigma_1=\epsilon$,
\item $B_1=\zer\vlor((x\vlan \zer)\vlor(\vneg w \vlor w))$
\item $\Psi_1$ is shown as part of Figure~\ref{figure:examplephase1}
\item $\tau_2=\esbs \zer {z_2}$,
\item $\rho_2=\eta_2$,
\item $\mu_2=\asbs{v_9 \vlor v_9}{v_9}$,
\item $\sigma_2=\epsilon$,
\item $B_2= \zer\vlor ((w \vlor\vneg w )\vlan ((w \vlor \vneg w )\vlan ((\vneg x\vlor x)\vlan (\vneg w\vlor w))))$
\item $\Psi_2$ is shown as part of Figure~\ref{figure:examplephase1}
\item $\tau_3=\esbs{ ((z_1\vlor z_2) \vlor ((\vneg w\vlan w)\vlor (((x\vlor\zer)\vlan\vneg x)\vlor(w\vlan\vneg w))))}{z_3}$,
\item $\rho_3=\eta_3$,
\item $\mu_3=\asbs{v_{13} \vlor v_{13}}{v_{13}}$,
\item $\sigma_3=\epsilon$,
\item $B_3= y_1\vlan y_2$,
\item $\Psi_3$ is shown as part of Figure~\ref{figure:examplephase1},
\end{itemize}

where $y_1,y_2$ are fresh variables. We compose these as described in Construction~\ref{cons:phase-1-b} to obtain the derivation $\derb$ in Figure~\ref{figure:examplephase1}

\begin{figure*}[]
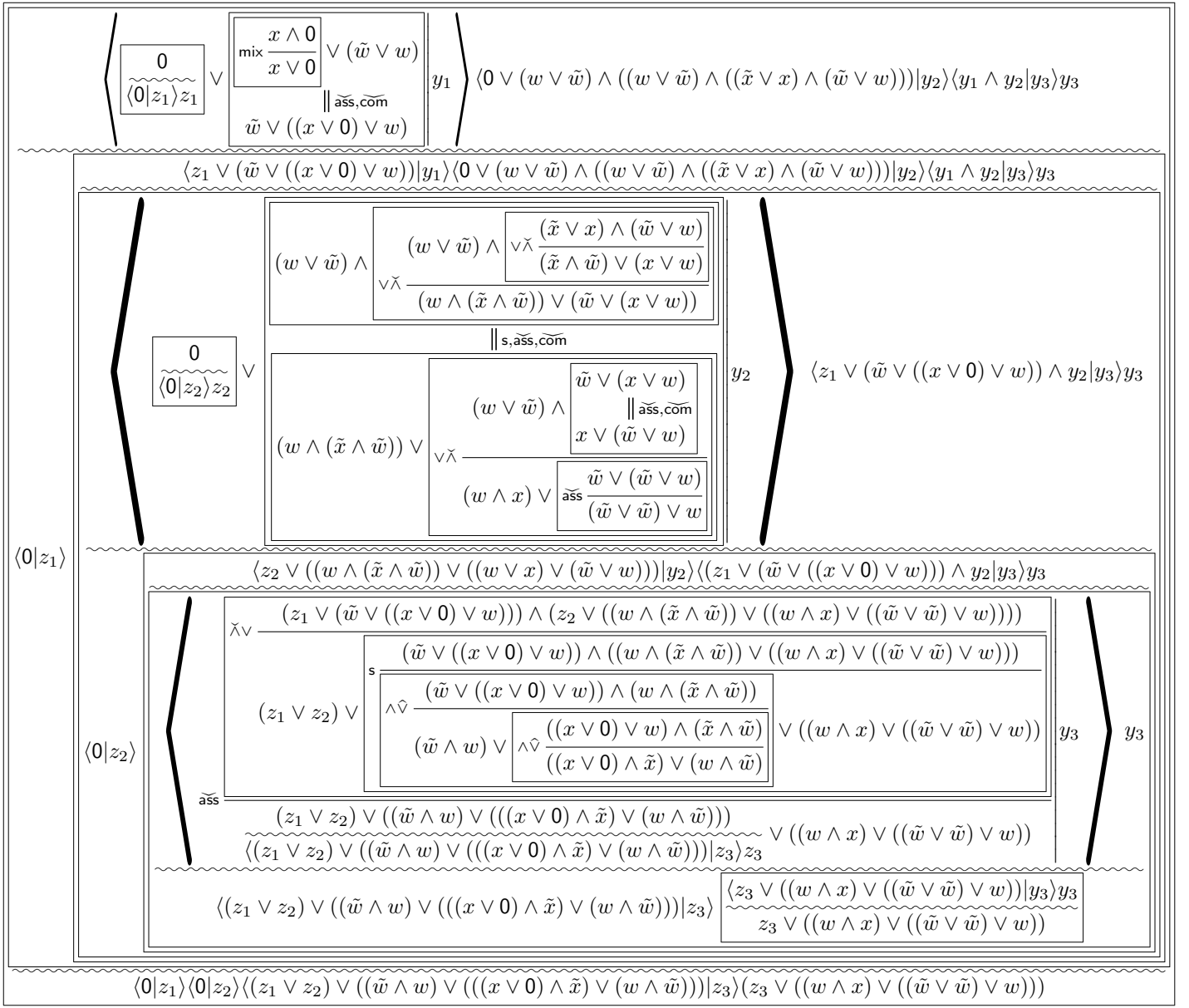

\[\od{\odq{\odq
  {\esbs{\odq{\zer}{}{\esbs{\zer}{z_1}z_1}{} \vlor \odv
  {[\odn{(x.\zer)}\mix{[x.\zer]}{}.
  (\vneg w\vlor w) ]}
{}{[\vneg w.[[x.\zer].w]]}{\assD,\comD}}{y_1}
\esbs{\zer\vlor ((w \vlor\vneg w )\vlan ((w \vlor \vneg w )\vlan ((\vneg x\vlor x)\vlan (\vneg w\vlor w))))}{y_2}
\esbs{y_1\vlan y_2}{y_3}y_3}
{}{\esbs{\zer}{z_1}\; \odq
  {\esbs{[z_1 .[\vneg w.[[x.\zer].w]]]}{y_1}
  \esbs{\zer\vlor ((w \vlor\vneg w )\vlan ((w \vlor \vneg w )\vlan ((\vneg x\vlor x)\vlan (\vneg w\vlor w))))}{y_2}
  \esbs{y_1\vlan y_2}{y_3}y_3}
{}{\odq
  {\esbs{\odq{\zer}{}{\esbs{\zer}{z_2}z_2}{} \vlor \,
\odv{\odbox{((w\vlor \vneg w).\odn{( (w\vlor\vneg w).\odn{( (\vneg x\vlor x).(\vneg w\vlor w) )}
{\vlor\sDn\vlan}{[(\vneg x.\vneg w).[ x.w]]}{})}
{\vlor\sDn\vlan}{[(w.(\vneg x.\vneg w)).[\vneg w.[ x.w]]]}{})}}
{}{\odbox{[(w.(\vneg x.\vneg w)).\odn
{([w.\vneg w].\odv
{[\vneg w.[ x.w]]}
{}{[ x.[\vneg w.w]]}{\assD,\comD})}
{\vlor\sDn\vlan}{[(w. x).\odn
{[\vneg w.[\vneg w.w]]}
{\assD}{[
{[\vneg w.\vneg w]}
.w]}{}]}{}]}}{\sw,\assD,\comD}}{y_2}
\esbs{[z_1 .[\vneg w.[[x.\zer].w]]]\vlan y_2}{y_3}y_3}
{}{\esbs{\zer}{z_2}\; \odq
  {\esbs{[z_2 .[(w . (\vneg x . \vneg w )) . [[w . x] . (\vneg w \vlor w)] ]]}{y_2}\esbs{([z_1 .[\vneg w.[[x.\zer].w]]].y_2)}{y_3}y_3}
{}{\odq
  {\esbs{\odframefalse\odn
{\odframetrue\odn
{ ([z_1.[\vneg w.[[x.\zer].w]] ]
  .[z_2.
    [(w.(\vneg x.\vneg w))
    .[(w.x) .[[\vneg w.\vneg w].w] ] ]])}
{\sDn\vlan\vlor}{[[z_1.z_2] 
   .\odn
   {([\vneg w.[[x.\zer].w]].[(w.(\vneg x.\vneg w)).[(w.x) .[[\vneg w.\vneg w].w] ] ] )}
   {\sw}{[\odn
        {([\vneg w.[[x.\zer].w]].(w.(\vneg x.\vneg w)) )}
        {\vlan\sUp\vlor}{[(\vneg w.w) .\odn
{([[x.\zer].w] .(\vneg x.\vneg w) ) }
{\vlan\sUp\vlor}{[([x.\zer] .\vneg x) . (w.\vneg w) ] }{} ] }{}. 
        [(w.x) .[[\vneg w.\vneg w].w] ] ]}{} ]}{}}
{\assD}{[\odq
{[[z_1.z_2] .[(\vneg w. w).[([x.\zer].\vneg x).(w.\vneg w)] ] ] }
{}{\esbs{[[z_1.z_2] .[(\vneg w. w).[([x.\zer].\vneg x).(w.\vneg w)] ] ]}{z_3}z_3 }{}
. [(w.x) .[[\vneg w.\vneg w].w] ] ] }{}}{y_3}y_3 }
{}{\esbs{ ((z_1\vlor z_2) \vlor ((\vneg w\vlan w)\vlor (((x\vlor\zer)\vlan\vneg x)\vlor(w\vlan\vneg w))))}{z_3}\; \odq
  {\esbs{[z_3 . [(w.x) .[[\vneg w.\vneg w].w] ]] }{y_3}y_3 }
{}{[z_3 . [(w.x) .[[\vneg w.\vneg w].w] ]] }{} }{} }{} }{}} {} }{}}
{}{\esbs{\zer}{z_1}\esbs \zer {z_2}\esbs{ ((z_1\vlor z_2) \vlor ((\vneg w\vlan w)\vlor (((x\vlor\zer)\vlan\vneg x)\vlor(w\vlan\vneg w))))}{z_3}  ([z_3 . [(w.x) .[[\vneg w.\vneg w].w] ]])}{}}
\]
\caption{An example of the derivation $\derb$ in Phase I of the p-simulation construction for the substitution Frege proof $\Phi$ shown in Example~\ref{example:fregeaamp} consisting of two axioms and an application of modus ponens.}
\label{figure:examplephase1}
\end{figure*}

Due to space limitations, we show in Figure~\ref{figure:examplephase2} an example of Phase II with the formula $z_3 \vlor (w \vlan x)$ instead of $\vls [z_3 . [(w.x) .[[\vneg w.\vneg w].w] ]]$. This construction ensures that when the guarded substitutions 
\[
\begin{array}{rl}
    \nu_a &=  \gsbss{\zer}{w}{\one}{\vneg w}{l_a}\gsbss{\one}{w}{\zer}{\vneg w}{r_a} \\
    \nu_b &= \gsbss{\zer}{x}{\one}{\vneg x}{l_b}\gsbss{\one}{x}{\zer}{\vneg x}{r_b}
\end{array}
\]
are applied then the resulting formula will be interpreted as $(a\vlan b)$.

\begin{sidewaysfigure}
\[\od{\odd{\odd{\ods{\odh
{\esbs{[z_3.(w.x)]}y
 \esbs{y\rng{l_a}\ba y\rng{r_a}}y
 \esbs{y\rng{l_b}\ba y\rng{r_b}}yy }}
{}{\esbs{\odframefalse\odn
     {([z_3\rng{l_a}.(w\rng{l_a}.x\rng{l_a})] \ba
         [z_3\rng{r_a}.(w\rng{r_a}.x\rng{r_a})])}
     {}{\odframetrue[(z_3\rng{l_a}\ba z_3\rng{r_a}).
         \odn{((w\rng{l_a}.x\rng{l_a}) \ba (w\rng{r_a}.x\rng{r_a}))}
           {}{[(w\rng{l_a}.x\rng{l_a})\ba (w\rng{r_a}.x\rng{r_a})] }{} ]}{}}y
   \esbs{y\rng{l_b}\ba y\rng{r_b}}yy
}{}}
{}{\esbs{z_3\rng{l_a}\ba z_3\rng{r_a}}{z_3}
   \esbs{w\rng{l_a}\ba w\rng{r_a} }{w}
   \esbs{x\rng{l_a}\ba x\rng{r_a}}{x}
   \esbs{\odframefalse\odn
     {([z_3\rng{l_b}.(w\rng{l_b}.x\rng{l_b})] \bb
         [z_3\rng{r_b}.(w\rng{r_b}.x\rng{r_b})])}
     {}{\odframetrue[(z_3\rng{l_b}\bb z_3\rng{r_b}).
         \odn{((w\rng{l_b}.x\rng{l_b}) \bb (w\rng{r_b}.x\rng{r_b}))}
           {}{[(w\rng{l_b}.x\rng{l_b})\bb (w\rng{r_b}.x\rng{r_b})] }{} ]}{}}yy
   }{}}
{}{[\esbs{z_3\rng{l_a}\ba z_3\rng{r_a}}{z_3}
    \esbs{z_3\rng{l_b}\bb z_3\rng{r_b}}{z_3}z_3. (
    \od{\ods{\ods{\odh
     {\esbs{w\rng{l_a}\ba w\rng{r_a}}{w}
    \esbs{w\rng{l_b}\bb x\rng{r_b}}{w}w}}
     {}{\esbs{\odframefalse\odn
       {((w\rng{l_a,l_b} \ba w\rng{r_a,l_b}))\bb 
        ((w\rng{l_a,r_b} \ba w\rng{r_a,r_b}))}
       {}{\odq{w\rng{l_a,l_b} \bb w\rng{l_a,r_b}}
       {}{\esbs{w\rng{l_a,l_b} \bb w\rng{l_a,r_b}}ww }{}\ba 
        \odq{w\rng{r_a,l_b} \bb w\rng{r_a,r_b}}
        {}{\esbs{w\rng{r_a,l_b} \bb w\rng{r_a,r_b}}ww }{}}{}}ww}{}}
     {}{( \esbs{w\rng{l_a,l_b} \bb w\rng{l_a,r_b}}ww
         \ba
          \esbs{w\rng{r_a,l_b} \bb w\rng{r_a,r_b}}ww) }{}}
    .
    \od{\ods{\ods{\odh
    {\esbs{x\rng{l_a}\ba x\rng{r_a} }x
     \esbs{x\rng{l_b}\ba x\rng{r_b} }xx }}
    {}{\esbs{
          \odq{x\rng{l_a,l_b}\ba x\rng{r_a,l_b} }
            {}{\esbs{x\rng{l_a,l_b}\ba x\rng{r_a,l_b}}xx }{}\bb
          \odq{x\rng{l_a,r_b}\ba x\rng{r_a,r_b}}
            {}{\esbs{x\rng{l_a,r_b}\ba x\rng{r_a,r_b}}xx}{} }xx}{}}
    {}{\esbs{x\rng{l_a,l_b}\ba x\rng{r_a,l_b}}xx \bb
       \esbs{x\rng{l_a,r_b}\ba x\rng{r_a,r_b}}xx}{}}   
    {}
    )] }{}}
\]
\caption{An example of phase II of the p-simulation construction on the formula $z_3 \vlor (w \vlan x)$. All unlabelled derivations are only refactorings of explicit substitutions. The first part realises Lemma~\ref{lemma:phase2eversion} and the second part Lemma~\ref{lemma:phase2merge}.  }
\label{figure:examplephase2}
\end{sidewaysfigure}

\end{example}

\subsection{Additional material for Section~\ref{sec:guarded}}

  In Figure~\ref{fig:pythagoras} we give a ``semi-formalization'' of the proof of the Pythagoras theorem discussed in Section~\ref{sec:guarded}, using guarded substitutions.


\end{document}